\documentclass[11pt]{llncs}

\usepackage{fullpage}

\usepackage{amsmath}
\usepackage[usenames,dvipsnames]{xcolor}
\usepackage[colorlinks,citecolor=blue,linkcolor=BrickRed]{hyperref}
\usepackage{makeidx} 
\usepackage{algorithm}
\usepackage{algorithmic}
\usepackage{graphicx,tipa}
\usepackage{lmodern,fix-cm}
\usepackage{arcs}
\usepackage{times}
\usepackage{amsfonts,latexsym,graphicx,epsfig,amssymb,color}
\usepackage{slashbox,multirow}
\usepackage{rotating}
\usepackage{verbatim}
\usepackage{gensymb}
\usepackage{mathrsfs}
\usepackage{enumerate}

\newtheorem{observation}[theorem]{Observation}
\newcommand{\reals}{\mathbb{R}}

\newcommand{\ignore}[1]{}

\newcommand{\rob}[1]{\mathscr R_{#1}}
\newcommand{\ben}[1]{\mathscr B_{#1}}
\newcommand{\algo}{\mathscr{A}}

\newcommand{\norm}[1]{\left\lVert#1\right\rVert}
\newcommand{\sinn}[1]{\sin \left({#1}\right)}
\newcommand{\coss}[1]{\cos \left({#1}\right)}

\newcommand{\ki}[1]{C_{#1}}
\newcommand{\ci}[1]{\mathcal{C}(#1)}
\newcommand{\li}[1]{\mathrm{line}(#1)}
\newcommand{\cycle}{\mathrm{cycle}}

\newcommand{\cost}{\mathcal{C}}
\newcommand{\se}{\mathcal{S}}
\newcommand{\ev}{\mathcal{E}}

\newcommand{\avg}[1]{\mathrm{Avg}\left(#1\right)}
\newcommand{\wrs}[1]{\mathrm{Wrs}\left(#1\right)}

\newcommand{\E}{\mathrm{E}}

\newcommand{\evac}{$_2$\textsc{Evac}$_{F2F}$}
\newcommand{\evacw}{$_2$\textsc{Evac}$_{F2F}^w$}

\begin{document}

%\title{Patrolling a Path Connecting a Set of Points
%with Non-Uniform Frequencies of Visits}

\title{Average Case - Worst Case Tradeoffs \\ for Evacuating 2 Robots from the Disk in the Face-to-Face Model
\thanks{
This is the full version of the paper, with the same title and authors, that was accepted in the 14th International Symposium on Algorithms and Experiments for Wireless Sensor Networks (ALGOSENSORS’18),  23--24 August 2018, Helsinki, Finland.}
}

\author{
Huda Chuangpishit%\inst{1}
\and
Konstantinos Georgiou%\inst{1}
\thanks{Research supported in part by NSERC Discovery Grant.}
\and
Preeti Sharma%\inst{1}
}

\institute{
Department of Mathematics, Ryerson University \\
350 Victoria St, Toronto, ON, M5B 2K3, Canada \\
\email{h.chuang,konstantinos,preeti.sharma@ryerson.ca}
}

\maketitle

\begin{abstract}
The problem of evacuating two robots from the disk in the face-to-face model was first introduced in~\cite{CGGKMP}, and extensively studied (along with many variations) ever since with respect to worst case analysis.
%cost $\wrs{\algo}$ of evacuation algorithms $\algo$. 
We initiate the study of the same problem with respect to average case analysis,
%cost $\avg{\algo}$, 
which is also equivalent to designing randomized algorithms for the problem. 
First we observe that algorithm $\ben{2}$ of~\cite{CGGKMP} with worst case cost $\wrs{\ben{2}}:=5.73906$ has average case cost $\avg{\ben{2}}:=5.1172$. Then we verify that none of the algorithms that induced worst case cost improvements in subsequent publications has better average case cost, hence concluding that our problem requires the invention of new algorithms. 
Then, we observe that a remarkable simple algorithm, $\ben{1}$, has very small average case cost $\avg{\ben{1}}:=1+\pi$, but very high worst case cost $\wrs{\ben{1}}:=1+2\pi$. 

Motivated by the above, we introduce constrained optimization problem \evacw, in which one is trying to minimize the average case cost of the evacuation algorithm given that the worst case cost does not exceed $w$. The problem is of special interest with respect to practical applications, since a common objective in search-and-rescue operations is to minimize the average completion time, given that a certain worst case threshold is not exceeded, e.g. for safety or limited energy reasons. 

Our main contribution is the design and analysis of families of new evacuation parameterized algorithms $\algo(p)$ which can solve \evacw, for every $w \in [\wrs{\ben{1}},\wrs{\ben{2}}]$. In particular, by letting parameter(s) $p$ vary, we obtain parametric curve  
$\left(
\avg{\algo(p)}, \wrs{\algo(p)}
\right)$
that induces a continuous and strictly decreasing function in the mean-worst case space, and whose endpoints are 
$\left(
\avg{\ben{1}}, \wrs{\ben{1}}
\right)$
,
$\left(
\avg{\ben{2}}, \wrs{\ben{2}}
\right)$.
Notably, the worst case analysis of the problem, since it's introduction, has been relying on technical numerical, computer-assisted, calculations, following tedious robots trajectories' analysis. Part of our contribution is a novel systematic procedure, which, given \textit{any evacuation algorithm}, can derive it's worst and average case performance in a clean and unified way. \\
%\end{abstract}
\end{abstract}

\section{Introduction}

Search problems are concerned with the exploration of a domain, aiming to identify the location of a hidden object. 
More particularly, in evacuation-type problems where the domain is the unit disk, introduced recently by Czyzowicz et al. in~\cite{CGGKMP}, a group of mobile agents collectively search for a hidden item (the exit) placed on the perimeter of the disk, attempting to expedite the time it takes for the last agent to evacuate, i.e. reach the exit. As it was the case in~\cite{CGGKMP}, as well as in a series of follow-up improvements and problem variations, the main objective was the design of evacuation algorithms that minimize the \textit{worst case performance}. 
In contrast, real-life search-and-rescue operations, in which current problems find applications, are mostly concerned with good \textit{average performance}. Keeping also in mind that, in realistic search tasks, mobile agents do not have unbounded resources and at the same time it is imperative that the search terminates successfully with probability 1, one is motivated to study average case - worst case trade-offs for evacuation search problems. 

In this direction, we initiate the study of the traditional evacuation problem first introduced in~\cite{CGGKMP} from the perspective of average case analysis, which in our case is equivalent to designing efficient randomized algorithms. 
More specifically, we introduce problem \evacw\ which, at a high level, asks for efficient evacuation algorithms that perform well on average, given that their worst case performance does not exceed $w$ (which can be thought as the maximum time robots can operate, e.g. due to energy restrictions). The problem seems particularly challenging given that the worst case performance analysis of all known evacuation algorithms require tedious analysis, tailored to robots' trajectories, and followed by intense, computer-assisted calculations, which are always numerical. Our results pertain to new families of evacuation algorithms, whose worst case performance analysis can be done rigorously, and whose average case analysis requires again intense computer-assisted calculations, achieving average case - worst case trade-offs for a wide spectrum of values. Our computer-assisted calculations rely on a novel theoretical and unified approach to compute the cost of \textit{any evacuation algorithm} and for \textit{any placement of the hidden item} without relying on tedious analysis specific to robots' trajectories . Equipped with these techniques, we also verify, somehow surprisingly, that the best evacuation algorithms known prior to this work, designed to perform well in the worst case, \textit{do not perform well} for \evacw, adding this way to the motivation of our problem.

\subsection{Related Work}

In search problems, mobile agents, commonly referred as robots, aim to locate efficiently a hidden item placed in some geometric domain. 
Numerous search-types problems have been introduced and studied since the 60's, 
when two seminal papers on probabilistic search,~\cite{beck1964linear} and \cite{bellman1963optimal},
were concerned with minimizing the {\em expected time} to locate the item. 
The number of search-type variants, along with the difficulty of the underlying mathematical problems and the elegance of many results soon gave rise to what is known nowadays as Search Theory. 
Many of the variants have been classified in surveys, e.g.~\cite{dobbie1968survey}~\cite{benkoski1991survey}, while a number of books provide a comprehensive study for similar problems, e.g. see \cite{stone1975theory,ahlswede1987search,alpern2002theory} and the most recent~\cite{Alpern2013}.

Search-type problems have also been studied under the perspective of exploration 
in \cite{AH00,AKS02,DKP91,HIKK01} by a single robot, and in \cite{Y98,T01,B05} by multiple robots. 
Terrain mapping has been the main search task even in problems where exploration is not the primary objective, e.g. \cite{K94,mitchell2000geometric,PY}. 
Numerous other search-type problems have been introduced and classified as hide-and-seek and pursuit-evasion games, e.g. see~\cite{chung2011search,FT08,lidbetter2013hide,nahin2012chases}. 
Overall the list of search-type problems is enormous, and having given a representative list above, in what follows we refer only to the most relevant ones. 

The perception of a search-type problem as an evacuation problem, from a theoretical perspective, appeared almost a decade ago, e.g. in~\cite{baumann2009earliest,FGK10}.
The problem we study here is a direct follow-up to the evacuation problem \evac\ (a search-type problem) first introduced in \cite{CGGKMP}, which included many variants based on the number of robots and the communication model between them. 
In the variant \evac\ which is relevant to our work, two robots start from the center of the unit disk, while an exit is hidden somewhere on the perimeter. The robots move at speed 1, their perception of their environment is restricted to their location and they can exchange information only by meeting. The goal is to minimize the worst case evacuation time, i.e. the time it takes the last robot to reach the exit, over all exit placements. The upper bound of 5.73906 in~\cite{CGGKMP} was later improved to 5.628 in \cite{CGKNOV}, and further to the currently best known 5.625 in \cite{Watten2017}, while the best lower bound known for the problem is 5.255 due to~\cite{CGKNOV}. 

Since the introduction of \evac\ in~\cite{CGGKMP}, a number of variants emerged, focusing on different geometric domains, different number of robots and robots' specifications, different communication models etc. Examples include evacuation from the disk with more than 1 exits in the wireless model~\cite{DBLP:conf/icdcn/CzyzowiczDGKM16},
evacuation of a group of robots on a line~\cite{Groupsearch} (generalization of the celebrated Cow-Path problem ~\cite{baezayates1993searching}), 
evacuation in the presence of faulty robots in a line~\cite{isaacCzyzowiczGKKNOS16}
and in a disk~\cite{georgioudiskfaulty2017},
evacuation with advice~\cite{georgiou2017searching}
while more recently evacuation with combinatorial requirements on the robots that need to evacuate, e.g. 
%\cite{GeorgiouKK16} \cite{GeorgiouKK17} \cite{CGKKKNOS18a} \cite{CGKKKNOS18b}.
\cite{GeorgiouKK16,GeorgiouKK17,CGKKKNOS18a,CGKKKNOS18b}.

\subsection{Outline of Our Results \& Paper Organization}

We initiate the study of evacuating 2 robots from the disk in the face-to-face model from an average case complexity perspective. In particular we introduce problem \evacw\ in which one tries to minimize the expected performance of randomized evacuation algorithms, subject to that the worst case performance does not exceed $w$. The problem is particularly challenging given that existing positive results, from a worst case complexity perspective, rely on tedious theoretical analysis tailored to algorithmic solutions, and supported by intense computer-assisted calculations. One of our main contributions is a unified and simple approach to quantify the performance of any evacuation algorithm and for any input. Equipped with this technique, we first verify that none of the previously known evacuation algorithms has good average case performance. Then, we introduce families of evacuation algorithms that have competitive average case performance, given that their worst case performance does not exceed $w$, for a wide range of $w$'s. Our results rely on rigorous and technical worst case performance analysis for the newly proposed algorithms. Building upon our new technique for efficiently evaluating the cost of evacuation algorithms for any input, we are able to numerically compute the average case performance of our algorithms, as well as to quantify formally the induced average case -worst case trade-offs.

In Section~\ref{sec: definition} we formally define \evacw\ and we give a high-level outline of the results we establish. 
Section~\ref{sec: evacuation times} contains one of our main contributions, which is a systematic process to compute the performance of any evacuation algorithm, given that robots' trajectories have convenient representations, described in Section~\ref{sec: trajectories description}. 
In Section~\ref{sec: benchmark algos} we analyze two benchmark algorithms for \evacw, as well as we motivate further the problem for certain values of $w$, among others showing, somehow surprisingly, that none of the previously proposed evacuation algorithms is efficient for our problem. 
Section~\ref{sec: new evac algos} describes our main contributions in the form of new families of evacuation algorithms. Then, in Section~\ref{sec: new evac algos wrs} we perform rigorous worst case analysis for all new algorithms and in Section~\ref{sec: new evac algos avg} we perform average case analysis, using our results from Section~\ref{sec: evacuation times} along with
 heavy computer-assisted calculations. In the same section, we also quantify formally all our results for \evacw. Finally, in Section~\ref{sec: conclusion} we conclude with some open problems. 
%Some of the proofs are omitted due to space limitations. A full version of the paper can be found .... ??????????????????????????????????????

%All omitted proofs of statements in the main body of the paper can be found in the Appendix~\ref{appendix: omitted proofs}.

\section{Preliminaries}

\subsection{Problem Definition \& Main Results}
\label{sec: definition}
In \evac, two searchers (robots) start from the center of the unit disk. Moving at maximum speed 1, the two robots can move anywhere on the plane. Somewhere on the perimeter of the disk there is a hidden object (exit) that can be located by any of the robots only if the robot is co-located with the exit. 

The two robots do not see each other from distance, neither can they exchange messages unless they meet (face-to-face model), but they can agree in advance on each other's trajectories. A \textit{feasible evacuation} algorithm is determined by  the trajectories of the robots, in which eventually both robots reach the exit. 
For simplicity, we also require, w.l.o.g. that eventually any robot stays idle.
For convenience, we think that the center of the unit disk lies at the origin $(0,0)$ of a Cartesian system, and we denote by $\cycle(x)$ the point $(\coss{x}, \sinn{x})$, which will be referred to as an instance of \evac\ when the exit is placed at $\cycle(x)$.
Given instance $\cycle(x)$, we define the \textit{evacuation time} $\cost(x)$ of the feasible evacuation algorithm as the time it takes the last robot to reach the exit. 

In this work we are concerned with determining tradeoffs between the worst case and the average case performance (of uniform placements of the exit) of evacuation algorithms for \evac. More specifically, we say that an evacuation algorithm $\algo$ with evacuation cost $\cost(x)$ on instance $\cycle(x)$ is \textit{$(a,w)$-efficient} if 
\begin{align*}
& \avg{\algo}:= \E_{x\in [0,2\pi)}[\cost(x)] \leq a, \\
& \wrs{\algo}:= \sup_{x\in [0,2\pi)} \{ \cost(x) \} \leq w.
\end{align*}
where the expectation is with respect to the uniform distribution over $[0,2\pi)$. 
Special to our problem is that $\avg{\algo}$ can also be interpreted as the expected performance of a randomized algorithm based on $\algo$. 
Indeed, consider an algorithm which first performs a random rotation of the disk around the origin of angle $\theta$, where $\theta$ is chosen uniformly at random from $[0,2\pi)$, and then simulates $\algo$. This random step is equivalent to choosing a deployment point uniformly at random on the disk. Due to the symmetry of the domain, it is irrelevant where the adversary will place the unique exit, and hence the expected performance of this randomized algorithm equals $\avg{\algo}$. 
%, in which the first robot to reach the perimeter of the circle chooses its deployment point uniformly at random, and thereafter $\algo$ is simulated. 

For algorithms $\algo(p)$ parameterized by parameter(s) $p$, the pair $\left(\avg{\algo(p)},\wrs{\algo(p)}\right)$ will correspond to a subset of $\reals^2$ (and a curve if $p$ is only one parameter), that we will refer to as the \textit{Efficient Frontier}. 
We also adopt an optimization perspective of the problem, and we introduce the following optimization problem \evacw\ on parameter $w$:

\begin{align}
\min~ & \frac1{2\pi} \int_0^{2\pi} \cost(x) dx  \tag{\evacw} \\
\textrm{s.t.} & ~~~~\cost(x) \leq w, ~~\forall x \in [0,2\pi). \notag
\end{align}

Due to an analysis we perform later, \evacw\ is interesting as long as 
$w_1\leq w\leq w_2$.
%where $w_1\approx 5.739, w_2\approx 7.283$.
At a high level, values $w_1,w_2$ above are obtained from two benchmark algorithms, $\ben{1}, \ben{2}$, where 
$\wrs{\ben{1}}=w_1\approx 5.739, \avg{\ben{1}}=a_1\approx 5.1172, \wrs{\ben{2}}=w_2\approx 7.283, \avg{\ben{1}}=a_2 \approx 7.28319$,
hence $\ben{1}$ being efficient in worst case and inefficient in average case,
while $\ben{2}$ being efficient in average case and inefficient in worst case.
As it is common for \evac\ (and many follow-up variation problems) closed forms for the cost of best-solutions known do not exist, and upper and lower bounds are given numerically. Our results involve upper bounds for a continuous spectrum of parameters $w$ for problem \evacw. In particular we propose families of algorithms $\algo$ (over some parameters) so that, as their parameters vary, we obtain $\wrs{\algo}=w$ and $\avg{\algo}=g(w)$, for each $w\in [w_1,w_2]$. The curve $(g(w),w)$ summarizing our results is depicted in Figure~\ref{fig: entire efficient frontier}, and it is later quantified in Theorem~\ref{thm: main thm upper bound} (see Section~\ref{sec: new evac algos wrs}). 

\begin{figure}[h!]
  \centering
  \includegraphics[width=0.7\linewidth]{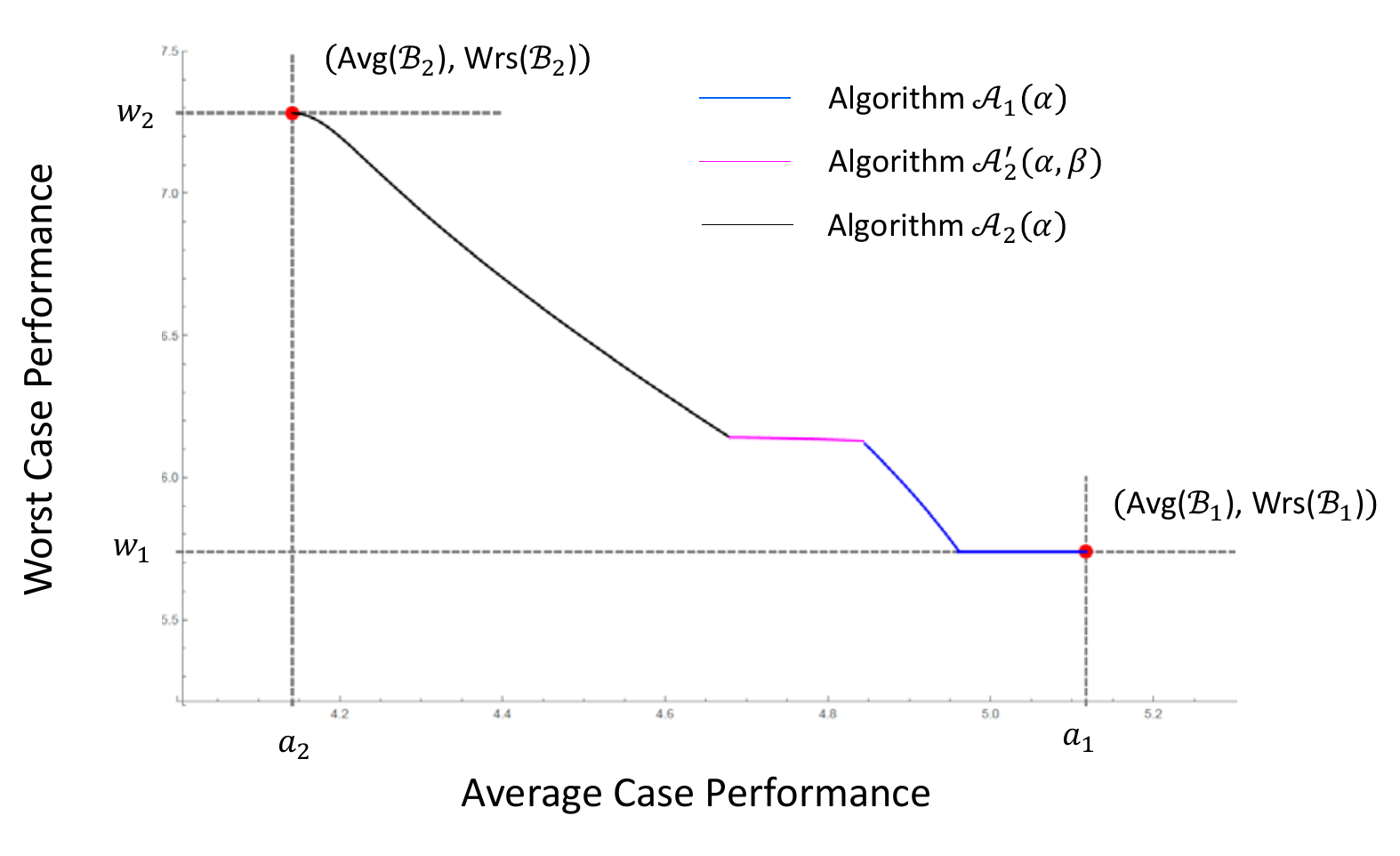}
\caption{
Illustration of the performance of our solution to \evacw, for every $w \in [w_1,w_2]$. Depicted curve corresponds to parametric curve $(g(w),w)$, where $w, g(w)$ are the worst case performance and average case performance of three different families of evacuation algorithms $\algo_1, \algo_2', \algo_2$, discussed formally in Section~\ref{sec: new evac algos}.
Note that the magenta curve is not a straight line and, as we show next, induces decreasing worst case performance (as the average  case performance increases).
}
\label{fig: entire efficient frontier}
\end{figure}

Note that an $(a,w)$-efficient algorithm gives a solution of value $a$ for \evacw. Our approach to prove Theorem~\ref{thm: main thm upper bound} is to define families of evacuations algorithms $\algo(p)$ parameterized by parameter(s) $p$. We will prove that these algorithms are $(u(p),v(p))$-efficient for some functions $u(p), v(p)$, and in particular the evaluation of the worst case performance will be exact and monotone in $p$, while the computation of $v(p)$ will be computer-assisted. Then we will set $p=v^{-1}(w)$, and will be able to describe the average case performance as a function of $w$ as $g(w):=u(v^{-1}(w))$.

\subsection{Computing Evacuation Times}
\label{sec: evacuation times}

For any feasible evacuation algorithm, we define by $\se(x)$, the first time that $\cycle(x)$ is visited by any robot. Clearly, when a robot, say $\rob{1}$ locates the exit at $\cycle(x)$, it may attempt to catch $\rob{2}$ while moving along $\rob{2}$'s trajectory along the shortest line segment, say of length $\ev(x)$. Once robots meet, they return together to $\cycle(x)$, inducing total evacuation cost
$
\cost(x) = 1+\se(x) + 2\ev(x). 
$

All existing results for \evac, from a worst case complexity perspective, rely on numerical computer-assisted estimation of $\sup_x \cost(x)$, after identifying properties of the maximizer. In this section, we elevate existing arguments, and we propose a generalized and unified approach for computing $\cost(x)$, for any $x$ and for any robots' trajectories. 
For the sake of formality, as well as for practical purposes, 
robots' trajectories will be defined by parametric functions $\mathcal{F}(t)=(f(t), g(t))$, where $f,g:\reals\mapsto \reals$ are continuous and piecewise differentiable. 
In particular, search protocols for the two robots will be given by trajectories 
$\rob{1}(t),\rob{2}(t)
$, 
where $\rob{i}(t)$ will denote the position of robot $\rob{i}$ at time $t\geq 0$. Therefore, any evacuation algorithm will be identified by a tuple $(\rob{1},\rob{2})$. To simplify notation, we will only determine the trajectories from the moment the two robots reach the perimeter of the circle, and until the entire circle is searched, and we will silently assume that robots stay put after exploration is over. 

\begin{lemma}
\label{lem: evac compute}
Consider instance $\cycle(x)$ of \evac, and suppose that for a feasible evacuation algorithm $(\rob{1},\rob{2})$, robot 1 is the first robot that finds the exit. Then $\ev(x) = \bar{t} - \se(x)$, where $\bar{t}=\bar{t}(x)$ is the smallest root, no less than $\se(x)$, of function 
\begin{equation}
\label{equa: evac as a root}
h_x(t):= 
\norm{
\rob{2}(t) - \rob{1}(\se(x))
} 
-t+\se(x).
\end{equation}
\end{lemma}

\begin{proof}
First observe that $h_x(t)$ is continuous, and assuming that the two robots are not co-located when the exit is found, we have $h_x(\se(x))>0$. At the same time, since the evacuation algorithm is feasible, $\rob{2}(t)$ is eventually a constant, and hence for big enough $t$ we have that $h_x(t)$ becomes eventually negative. By the mean value theorem, there is $t_0>0$ for which $h_x(t_0)=0$. 

Now consider the smallest positive root $\bar{t}$ of $h_x$, no less than $\se(x)$. At time $\bar{t}$,  $\rob{2}$ is located at point $\rob{2}(\bar{t})$, and it is 
$\norm{
\rob{2}(\bar t) - \rob{1}(\se(x))
}$ 
away from the location $\cycle(x)$ of the discovered exit. At the same time, $\rob{1}$ moves with speed 1 along the shortest path to catch $\rob{2}$ in her trajectory. Hence  it takes $\rob{1}$ some $\bar{t} - \se(x)$ extra time from the moment the exit is found until she reaches point $\rob{2}(\bar{t})$. By definition we have $\rob{1}(\bar{t})=\rob{2}(\bar{t})$, and therefore $\ev(x) = \bar{t} - \se(x)$ as claimed. 
\qed
\end{proof}

For some special trajectories, $\ev(x)$ admits a simpler description that we describe next. Before that, we introduce some notation pertaining to a function $\delta:[0,\pi]\mapsto \reals_+$, which we widely use in the remaining of the paper:
\begin{equation}
\label{equa: def of d_a}
\mathcal \delta(x) := \textrm{ unique non-negative root (w.r.t. $d$) of~~~`` } 2\sinn{x+\frac{d}{2}} = d \textrm{ ''.}
\end{equation}
To simplify notation, we will also abbreviate $\delta(x)$ by $\delta_x$. The fact that $\delta_x$ is well defined follows easily from the monotonicity of $\sin$ in $[0,\pi]$.

\begin{lemma}
\label{lem: simple catching time}
For some instance $\cycle(x)$ of \evac, suppose that for a feasible evacuation algorithm $(\rob{1},\rob{2})$, $\rob{1}$ is the founder of the exit, say at time $t_0=\se(x)$. Assume that both $\rob{1}(t_0), \rob{2}(t_0)$ lie on the circle at arc distance $2\alpha$, and suppose that $\rob{2}$'s movement is along the perimeter of the circle toward the complementary arc of length $2\pi-\alpha$. Then, 
$\ev(x)=\delta_\alpha$.
\end{lemma}

\begin{proof}
The lemma follows by applying transformation $t-\se(x)=d$ in the definition of $h_x(t)$ in 
Lemma~\ref{lem: evac compute}, so that $\ev(x) = t-\se(x)=d$. 
\qed
\end{proof}

We are ready to conclude with a corollary that will be handy for computing evacuation times numerically, and without relying on excessive case analysis, as it was the case before. 

\begin{corollary}
\label{cor: cost simplified}
Consider feasible evacuation algorithm $(\rob{1},\rob{2})$ for \evac. For any instance $\cycle(x)$ for which $\rob{1}$ is the exit founder, the evacuation cost can be computed as 
$
\cost(x) = 1+2\bar{t} -\se(x),
$
where $\bar{t}=\bar{t}(x)$ is the smallest root, at least $\se(x)$, of 
$h_x(t):= 
\norm{
\rob{2}(t) - \rob{1}(\se(x))
} 
-t
+\se(x)
$.
\end{corollary}

\subsection{Trajectories' Description}
\label{sec: trajectories description}

Robots' trajectories will be described in phases.
We will always omit the ``deployment phase'', i.e. the movement from the circle center to its perimeter, and we will only describe the trajectories from the moment robots start searching the circle. 
In each phase, robot $\rob{}$, will be moving between two explicit points, either along an arc, or along a line segment (chord of an arc), see Observations~\ref{obs: move circle} and~\ref{obs: move line} below.
We will summarize robot's trajectories in tables of the following format. 

\begin{center}
\begin{tabular}{l || clc}
\textit{Robot}		 	&	Phase \# &  \textit{Trajectory} & \textit{Duration}\\
\hline 
$\rob{}$	&	1					&	$\rob{}(t)$ & $t_1$ \\
				&	2					&	$\rob{}(t)$ &  $t_2$ \\
				&	$\vdots$				&	&  $\vdots$	\\
\hline			
\end{tabular}
\end{center}

In order to ease notation, trajectory $\rob{}(t)$ of phase $i$ will be described with parametric equations as if the time is reset to 0 after time $t_0 + t_1+t_2+ \ldots +t_{i-1}$, where $t_0=1$ (this is the time that robots reach the circle). The two fundamental trajectory components are movements along arcs and movements along line segments. 

\begin{observation}\label{obs: move circle}
Let $b \in [0,2\pi)$ and $\sigma \in \{-1,1\}$. The trajectory of an object moving at speed 1 on the perimeter of a unit circle with initial location $\cycle(b)$ is given by the parametric equation 
%$$ \ci{b,t}:=(\coss{\sigma t+b}, \sinn{\sigma t+b})$$
$
\cycle(\sigma t+b)=(\coss{\sigma t+b}, \sinn{\sigma t+b}).
$
%\ki{\sigma t+b}.$$ 
If $\sigma=1$ the movement is counter-clockwise (ccw), and clockwise (cw) otherwise. 
\end{observation}

\ignore{
\begin{proof}
Clearly, $\ci{b,0}=\ki{b}$. Also, it is easy to see that $\norm{\ci{b,t}}=1$, i.e. the object is moving on the perimeter of the unit circle. Lastly, 
$$
\left(\frac{d}{dt} \coss{\sigma t+b}\right)^2 + 
\left(\frac{d}{dt} \sinn{\sigma t+b}\right)^2 
=
\sigma^2 \left(- \sinn{\sigma t+b}\right)^2 + 
\sigma^2 \left( \coss{\sigma t+b}\right)^2 =1, 
$$
so the claim follows by Lemma~\ref{lem: unit speed}. 
\qed
\end{proof}
}

\begin{observation}\label{obs: move line}
Consider distinct points $A=(a_1,a_2), B=(b_1,b_2)$ in $\reals^2$. The trajectory of a speed 1 object moving along the line passing through $A,B$ and with initial position $A$ is given by the parametric equation 
$
\li{A,B,t}:=\left(
\frac{b_1-a_1}{\norm{A-B}}t+a_1,
\frac{b_2-a_2}{\norm{A-B}}t+a_2
\right).
$
\end{observation}

\ignore{
\begin{proof}
It is immediate that the parametric equation corresponds to a line. Also, it is easy to see that $\li{A,B,0}=A$ and $\li{A,B,\norm{A-B}}=B$, i.e. the object starts from $A$, and eventually visits $B$. As for the object's speed, we calculate
$$
\left(\frac{d}{dt} \left(\frac{b_1-a_1}{\norm{A-B}}t+a_1\right)\right)^2 + 
\left(\frac{d}{dt} \left(
\frac{b_2-a_2}{\norm{A-B}}t+a_2
\right)
\right)^2 
=
\left(
\frac{b_1-a_1}{\norm{A-B}}
\right)^2 + 
\left(
\frac{b_2-a_2}{\norm{A-B}}
\right)^2 =1
$$
so, by Lemma~\ref{lem: unit speed}, the speed is indeed 1. 
\qed
\end{proof}
}

Finally, the analysis of our algorithms' trajectories will give rise to a number of constants.
For the reader's convenience, we list here the numerical values of the most common constants that will be encountered later;
$
w_1 \approx 5.73906,
w_0 \approx 6.11953,
w' \approx 6.12851,
w_2 \approx 7.28319,
\alpha' \approx 1.15468,
\bar \alpha \approx 1.54419,
\beta' \approx 0.0241653,
\beta_0 \approx 0.04388.
$
All constants are formally defined when they are first introduced.

\section{Two Benchmark Algorithms \& Motivation}
\label{sec: benchmark algos}
In this section we describe two benchmark algorithms for \evac, as well as perform average case analysis to algorithms previously proposed in the literature. The reader may consult Figure~\ref{fig: benchmark} for the algorithms analyzed in this section. 
\begin{figure}[h!]
  \centering
  \includegraphics[width=1.0\linewidth]{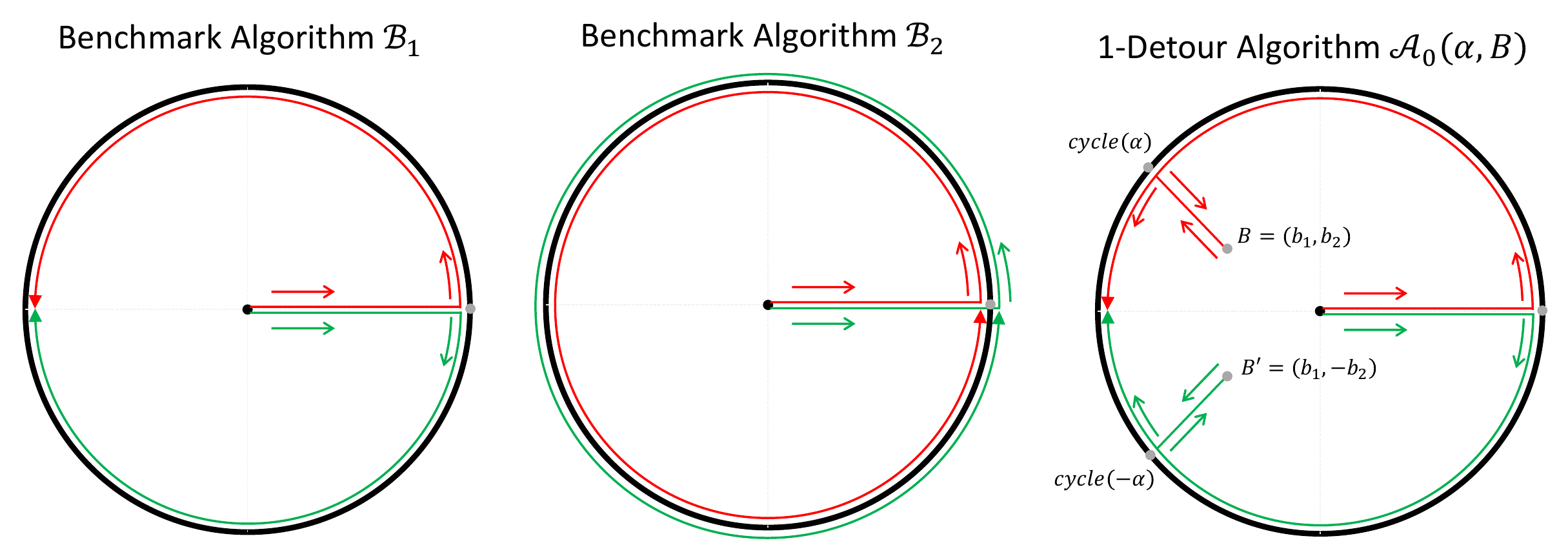}
\caption{Robots' Trajectories for algorithms $\ben{1}, \ben{2}, \algo_0$. The depicted trajectories show the search of the circle, and not the evacuation step that is performed once the exit is found.
}
%Description here ...... 
%Say that we show the trajectories of searching and that evacuation is not depicted. 
%Elaborate a little more in the main body?
%Say also that for the parts that are searched by both robots, the robots are co-located.}
\label{fig: benchmark}
\end{figure}
Czyzowicz et al.~\cite{CGGKMP} were the first to introduce an evacuation algorithm for \evac, which we denote here by $\ben{1}$ (see Figure~\ref{fig: benchmark} on the left). 

\begin{definition}[Benchmark Algorithm $\ben{1}$]
For all $t \in [0,\pi]$,  $\rob{1}(t)=cycle(t)$ and $\rob{2}(t)=cycle(-t)$.
\end{definition}

\begin{observation}
\label{obs: ben1 performance}
Benchmark Algorithm $\ben{1}$ is $(5.1172,5.73906)$-efficient.
\end{observation}

%\ignore{
\begin{proof}[Observation~\ref{obs: ben1 performance}]
Note that it takes time $\pi$ to search the entire circle, and that the two trajectories are symmetric with respect to horizontal axis. Therefore, we may assume that the instance $\cycle(x)$ satisfies $x\in [0,\pi]$. 

Clearly, for any such $x$, we have that $\se(x)=x$. 
By Lemma~\ref{lem: simple catching time}, we have that 
$\cost(x)=1+\se(x)+2\ev(x) =1+x+2\delta_x$. 
Numerical calculations (software assisted) show that 
\begin{align*}
&\wrs{\ben{1}}= \sup_{x\in [0,\pi]} \{ \cost(x) \} = 
 \sup_{x\in [0,\pi]} \{ 1+x+2\delta_x  \} \approx 5.73906 
, \\
&\avg{\ben{1}} = \E_{x\in [0,\pi]}[\cost(x)] 
=\frac1{\pi} \int_{x=0}^{\pi}\left(1+x+2\delta_x \right)dx \approx 5.1172.
\end{align*}
\qed
\end{proof}
%}

$\ben{1}$ should be understood as being efficient in the worst case, but inefficient on average. 
The claim becomes transparent by introducing the following \textit{naive} algorithm for \evac\ that we depict in the middle of Figure~\ref{fig: benchmark}. 
\begin{definition}[Benchmark Algorithm $\ben{2}$]
For each $t\in [0, 2\pi]$, $\rob{1}(t)=\rob{2}(t)=cycle(t)$.
\end{definition}

\begin{observation}
\label{obs: ben2 performance}
Benchmark Algorithm $\ben{2}$ is $(1+\pi, 1+2\pi)$-efficient.
\end{observation}

%\ignore{
\begin{proof}[Observation~\ref{obs: ben2 performance}]
It is easy to see that for all $x\in [0,2\pi)$ we have $\bar{t}(x)=\se(x)=x$ and $\ev(x)=0$. Therefore $\cost(x)=1+x$, and hence
\begin{align*}
&\wrs{\ben{2}} = \sup_{x\in [0,2\pi)} \{ \cost(x) \} = 1+2\pi, \\
&\avg{\ben{2}} = \E_{x\in [0,2\pi)}[\cost(x)] = \int_{x=0}^{2\pi}\left(1+x\right)dx = 1+\pi.
\end{align*}
\qed
\end{proof}
%}

$\ben{2}$ should be understood as highly efficient on average, but inefficient in the worst case. Moreover, it should be clear that $\ben{1}, \ben{2}$ are feasible solutions to \evacw, for $w=5.1172$ and $w=1+2\pi$, respectively. 
We conjecture that $\ben{1}$ is indeed the optimal evacuation algorithm among all algorithms with worst case performance no more than $1+2\pi$. At the same time, below we show that $\ben{2}$ is the best algorithm for \evacw, when $w=5.1172$, among those previously used to improve upon the worst case performance. 
The importance of this observation is twofold; first we are motivated to study \evacw\ for the entire spectrum of $w\in [\wrs{\ben{1}},\wrs{\ben{2}}]$, and second we deduce that in order to perform well on average, we need to devise and analyze new evacuation algorithms. 

\ignore{
Algorithm $\ben{1}$ has bad average case performance, but competitive worst case performance. 
In contrast, Algorithm $\ben{2}$ performance should be understood as highly competitive on average but bad in the worst case. In fact, we conjecture the following. 
\begin{conjecture}
Every feasible evacuation algorithm for \evac\ has average case performance at least $1+\pi$. 
\end{conjecture}
}

Upper bounds for the  worst case performance of $\ben{1}$ were later improved in~\cite{CzyzowiczGKNOV15,Watten2017}, first to  5.628, and then to  5.625, using refined algorithms, respectively. The main idea behind the improvement is to understand the monoticity of $\cost(x)$ for algorithm $\ben{1}$. Indeed, the following lemma was implicit in both~\cite{CzyzowiczGKNOV15,Watten2017}, and can be  obtained numerically. 
\begin{lemma}\label{lem: monotonicity ben1}
There is $\alpha_0$, where $\alpha_0 \approx 0.96782$, so that evacuation cost $\cost(x)$ of $\ben{1}$ for \evac\ on instance $\cycle(x)$ is strictly increasing for $x\in [0,\alpha_0]$, and strictly decreasing in $x\in [\alpha_0, \pi]$. In particular, $\wrs{\ben{1}}=\cost(\alpha_0) \approx 5.73906$. 
\end{lemma}

\ignore{
h[x_] := y /. FindRoot[y == 2*Sin[x + y/2], {y, 1}]
Plot[ x + 2 h[x], {x, 0.967823, 0.967826}]
}

Consider now an execution of $\ben{1}$ in which one of the robots, say $\rob{2}$ continues searching on the circle and is close to approach a location that would be the meeting point if the instance was $\cycle(\alpha_0)$. In an attempt to help expedite a potential meeting (in case $\rob{1}$ is approaching) and effectively reducing the cost of the worst case, $\rob{2}$ would make a minor detour toward the interior of the disk, before returning back to the exploration of the circle. This simple idea was explored in~\cite{CzyzowiczGKNOV15} where the following family of algorithms were introduced, parameterized by $\alpha \in [0,\pi] $ and point $B$ within the unit disk, see also right of Figure~\ref{fig: benchmark}.

\begin{definition}[1-Detour Algorithm $\algo_0(\alpha,B)$]
\label{def: 1 detour algo}
For all $t \in [0,\pi+2\norm{\cycle(\alpha)-B}]$, the trajectory of $\rob{1}$ is defined as
\begin{center}
\begin{tabular}{l || lllc}
\textit{Robot}		 	&	Phase \#  & \textit{Trajectory} & \textit{Duration}\\
\hline 
$\rob{1}$	&	1					&	$\cycle(t)$ & $\alpha$ \\
				&	2					&	$\li{\cycle(\alpha),B,t}$ &  $\norm{\cycle(\alpha)-B}$ \\
				&	3					&	$\li{B,\cycle(\alpha),t}$ &  $\norm{\cycle(\alpha)-B}$ \\
	&	4			&			$\cycle(t+\alpha)$ & $\pi-\alpha$ \\
\hline			
\end{tabular}
\end{center}

\noindent The trajectory of $\rob{2}$ is symmetric with respect to the horizontal axis. 
\end{definition}

The crux of the contribution of~\cite{CzyzowiczGKNOV15} was to prove that there exists $\alpha,B$ for which the worst case performance is no more than 5.644 (and a delicate refinement is needed to achieve 5.628). Notably, their analysis is tedious and lengthy, whereas we can obtain the same result, relying again on numerical calculations, with minimal effort. Then, \cite{Watten2017} introduced variations of $\algo_0(\alpha,B)$ in which each robot performs more than 1 detours (see Phases 2,3 of $\algo_0(\alpha,B)$). Hence, $t$-detour algorithms are parameterized by a sequence $\alpha_1, \ldots, \alpha_t$, where $\alpha_i\geq 0$ and $\sum_i \alpha_i \leq \pi$, and points $B_i$ in the disk. Even 2-detour algorithms achieve worst case performance 5.625, while for each $t\geq 2$, t-detour algorithms do induce strictly improved performance (for appropriate choices of the parameters) but the improvement is negligible. 

Motivated by the results in~\cite{CzyzowiczGKNOV15,Watten2017}, one is tempted to ask whether any algorithm in the family $\algo_0(\alpha,B)$ improves upon $\ben{1}$ with respect to the average case analysis. The next claim is due to exhaustive, computer-assisted numerical calculations, see also Figure~\ref{fig: 1detourEF}. 

\begin{theorem}
\label{thm: 1 detour efficient frontier}
For every $\alpha \in [0,\pi)$ and for every $B$ in the unit disk
$
\avg{\algo_0(\alpha,B)}\geq \avg{\ben{1}}
$.
\end{theorem}

Theorem~\ref{thm: 1 detour efficient frontier} provides strong motivation for studying problem \evacw, since it shows that in oder to establish good upper bounds, i.e. our main results depicted in Figure~\ref{fig: entire efficient frontier} and quantified later in Theorem~\ref{thm: main thm upper bound}, one needs to employ new evacuation algorithms. 
Recall that even $\wrs{\ben{1}}$ that was first calculated in~\cite{CGGKMP}, or $\wrs{\algo_0(\alpha,B)}$ first calculated in \cite{CzyzowiczGKNOV15} for various $\alpha,B$, were all estimated with computer-assisted calculations. 
Due to the nature of the problem, we are bound to rely on computer-assisted calculations as well. Notably, our much more intense computational work is feasible only because we employ the new method for computing evacuation times due to Corollary~\ref{cor: cost simplified} and Definition~\ref{def: 1 detour algo} of $\algo_0(\alpha,B)$ trajectories. Overall, in order to verify~Theorem~\ref{thm: 1 detour efficient frontier} we compute pairs 
$\left( \avg{\algo_0(\alpha,B)}, \wrs{\algo_0(\alpha,B)} \right)$ for more than 500,000 different parameter values and we depict them in Figure~\ref{fig: 1detourEF}. 

%see algo 5 - 2018-04-04-PaperPicture.nb
\begin{figure}[h!]
  \centering
  \includegraphics[width=0.45\linewidth]{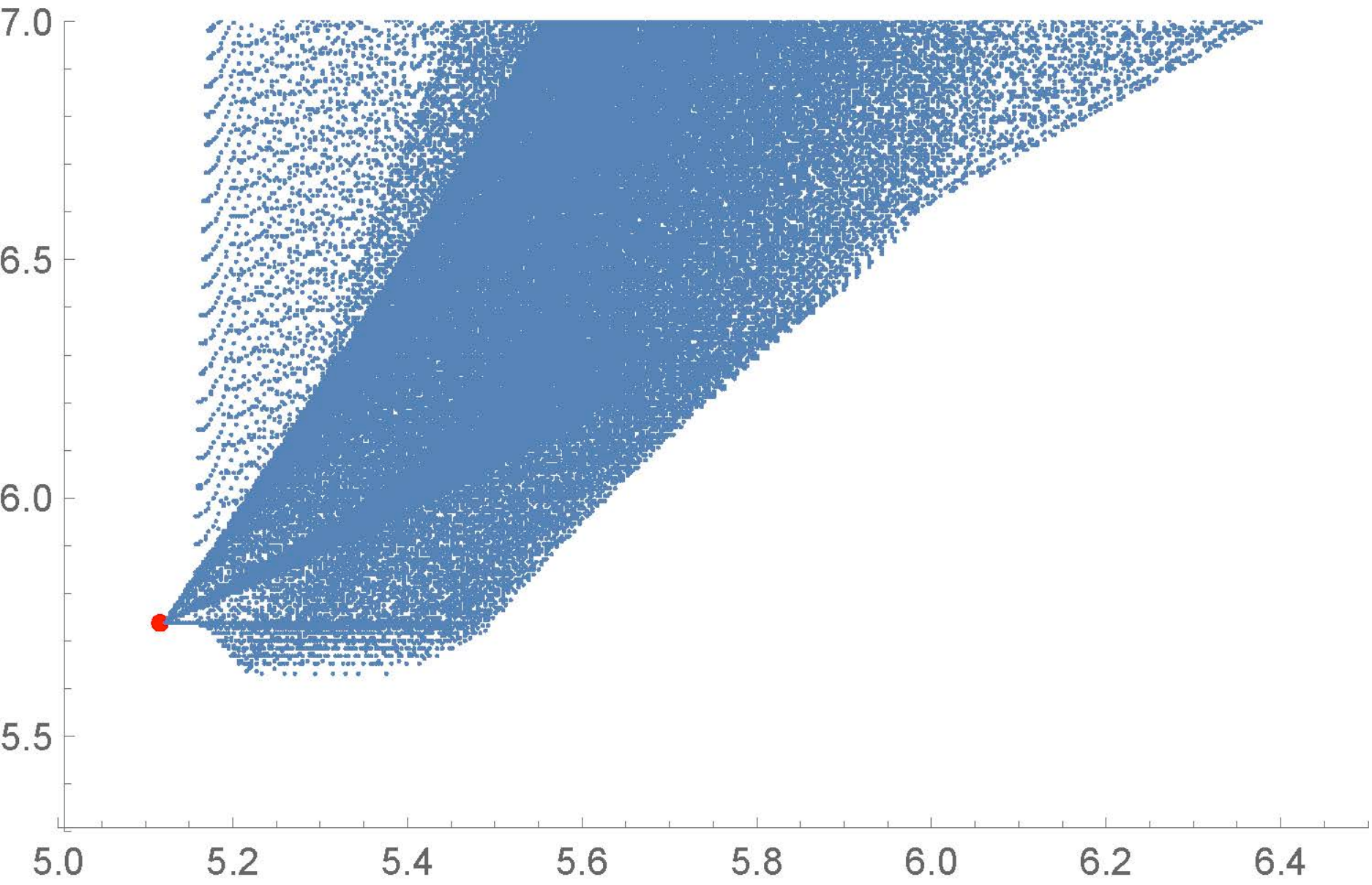}
\caption{Performance analysis of $\algo_0(\alpha,B)$ for various values of parameters $\alpha,B$. Blue points $(a,w)$ correspond to $(a,w)$-efficient algorithms $\algo_0(\alpha,B)$. The red point is $\left(\avg{\ben{1}}, \wrs{\ben{1}}\right)$, i.e. the performance of $\ben{1}$ in the average-worst case space. Note that no algorithm $\algo_0$ performs better on average than $\ben{1}$, while all  $\algo_0(t,\cycle(t))$ is exactly $\ben{1}$ for every point $t\in [0,\pi]$. 
%No algorithm performs better than $\ben{1}$ on average, while those that 
 }
\label{fig: 1detourEF}
\end{figure}

\section{New Evacuation Algorithms}
\label{sec: new evac algos}

In this section we propose families of evacuation algorithms for problem \evacw, for the entire spectrum of
$w\in [\wrs{\ben{1}},\wrs{\ben{2}}]$. Our algorithms are summarized in Figure~\ref{fig: newalgos}. 
\begin{figure}[h!]
  \centering
  \includegraphics[width=1.0\linewidth]{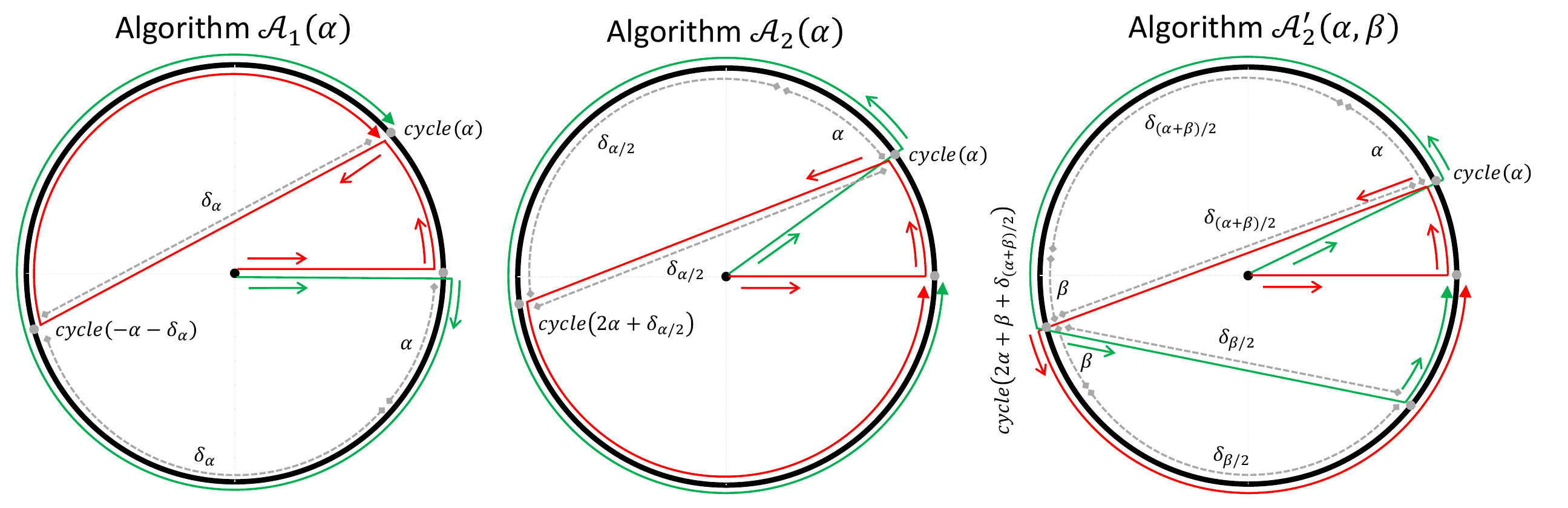}
\caption{
Robots' Trajectories for algorithms $\algo_{1},\algo_{2}, \algo_{2}'$. 
The depicted trajectories show the search of the circle, and not the evacuation step that is performed once the exit is found.
Arcs that are searched by both robots are also searched simultaneously, i.e. robots are co-located and search together.
}
\label{fig: newalgos}
\end{figure}

%algo1 is what appears as algo7a in 2018-02-27.pdf algo near w1						algo near w2 need better analysis
%algo2prime is what appears as algo3 2 jumps in the mathematica folder. 		
%algo2 is what appears as algo3 in 2018-02-27.pdf										

First we define families of evacuation algorithms that, as we show next, perform well for \evacw\ in the ``neighborhood of $\ben{1}$'', i.e. for $w$ close to $\wrs{\ben{1}}$. Our algorithms are parameterized by $\alpha$, and their circle exploration lasts $2\pi-\alpha$. 

\begin{definition}[Algorithm $\algo_{1}(\alpha)$]
\label{def: algo 1 new}
For all $t \in [0,2\pi-\alpha]$, the trajectory of $\rob{1}$ is defined as
\begin{center}
\begin{tabular}{l || lllc}
\textit{Robot}		 	&	Phase \#  & \textit{Trajectory} & \textit{Duration}\\
\hline 
$\rob{1}$	&	1					&	$\cycle(t)$ & $\alpha$ \\
				&	2					&	$\li{\cycle(\alpha),\cycle(-\alpha - \delta_\alpha),t}$ &  $\delta_\alpha$ \\
				&	3					&	$\cycle(-\alpha-\delta_\alpha -t)$ & $2\pi-2\alpha-\delta_\alpha$ \\
\hline			
\end{tabular}
\end{center}
\noindent where $\delta_a$ is defined in~\eqref{equa: def of d_a}. 
The trajectory of $\rob{2}$ is defined as $\rob{2}(t)=\cycle(-t)$, for all $t \in [0,2\pi-\alpha]$.
\end{definition}

$\algo_{1}$ is depicted in Figure~\ref{fig: newalgos} on the left. At a high level $\algo_{1}(\alpha)$ is a modification of $\ben{1}$ that is based on the following idea. 
The execution of $\algo_1(\alpha)$ is the same as in $\ben{1}$ until each robot searches an arc of length $\alpha$ (and hence $\algo(\pi)$ coincides with $\ben{1}$). 
After time $\alpha$, $\rob{1}$ abandons her trajectory and catches $\rob{2}$, on the perimeter of the circle resembling a trajectory as if the exit was located at $\rob{1}(\alpha)$. It is not difficult to see that the definition of $\delta_\alpha$ above satisfies $\rob{1}(\alpha+\delta_\alpha)=\rob{2}(\alpha+\delta_\alpha)=\cycle(-\alpha-\delta_\alpha)$.

Next we define a family of algorithms $\algo_2$ which, as we show later, perform well in the ``neighborhood of $\ben{2}$'',  i.e. for $w$ close to $\wrs{\ben{2}}$.
For this recall definition~\eqref{equa: def of d_a} of $\delta_a$. We let $\gamma_0 \approx 2.2412$
be the root of $2\alpha+\delta_{\alpha/2} = 2\pi$. For every $\alpha \leq \gamma_0$ we define a family of algorithms on parameter $\alpha$ whose circle exploration lasts $2\pi-\alpha$. 

\begin{definition}[Algorithm $\algo_{2}(\alpha)$]
\label{def: algo 2 new}
For all $t \in [0,2\pi-\alpha]$, the trajectory of $\rob{1}$ is defined as
\begin{center}
\begin{tabular}{l || lllc}
\textit{Robot}		 	&	Phase \#  & \textit{Trajectory} & \textit{Duration}\\
\hline 
$\rob{1}$	&	1					&	$\cycle(t)$ & $\alpha$ \\
				&	2					&	$\li{\cycle(\alpha),\cycle(2\alpha + \delta_{\alpha/2}),t}$ &  $\delta_{\alpha/2}$ \\
				&	3					&	$\cycle(2\alpha+\delta_{\alpha/2}+t)$ & $2\pi-2\alpha-\delta_{\alpha/2}$ \\
\hline			
\end{tabular}
\end{center}
\noindent The trajectory of $\rob{2}$ is defined as $\rob{2}(t)=\cycle(\alpha+t)$, for all $t \in [0,2\pi-\alpha]$.
\end{definition}

$\algo_2$ is depicted in the middle of Figure~\ref{fig: newalgos}. 
The condition that $\alpha\leq \gamma_0$ is added for simplicity to ensure that the latest catching point occurs while the other robot is still searching, and is not mandatory. 
At a high level $\algo_{2}(\alpha)$ is a generalization of $\ben{2}$ (note that $\algo_{2}(0)=\ben{2}$). For the first $\alpha$ time units, robots search in the same direction until $\rob{1}$ arrives at the deployment point of $\rob{2}$. Then, $\rob{1}$ catches $\rob{2}$ on the circle, as if the exit was located at $\rob{1}(\alpha)$ (which by Lemma~\ref{lem: simple catching time} happens in $\delta_{\alpha/2}$ extra time).

Finally we introduce a family of evacuation algorithms which will perform well for \evacw\ for intermediate values of $w\in [\wrs{\ben{1}},\wrs{\ben{2}}]$. 
For this we generalize family $\algo_{2}$ so that the two robots perform two alternating jumps, with parameters $\alpha, \beta$ satisfying 
$2\alpha+2\beta + \delta_{(\alpha+\beta)/2}+ \delta_{\beta/2}\leq 2\pi$, see right of Figure~\ref{fig: newalgos}.

\begin{definition}[Algorithm $\algo_{2}'(\alpha,\beta)$]
\label{def: algo 2 prime new}
For notational convenience, we set $\zeta_{\alpha,\beta}:=2\alpha + \beta + \delta_{(\alpha+\beta)/2}$. 
For all $t \in [0,2\pi-\alpha-\beta]$, the trajectories of $\rob{1}, \rob{2}$ are defined as follows
\begin{center}
\begin{tabular}{l || lllc}
\textit{Robot}		 	&	Phase \#  & \textit{Trajectory} & \textit{Duration}\\
\hline 
$\rob{1}$	&	1					&	$\cycle(t)$ 																				& $\alpha$ \\
				&	2					&	$\li{\cycle(\alpha),\cycle\left(\zeta_{\alpha,\beta}\right),t}$ 			&  $\delta_{(\alpha+\beta)/2}$ \\
				&	3					&	$\cycle\left(\zeta_{\alpha,\beta}+t\right)$ 									& $2\pi-2\alpha-\beta - \delta_{(\alpha+\beta)/2}$ \\
\hline			
\ignore{
\end{tabular}
\end{center}
\begin{center}
\begin{tabular}{l || lllc}
\textit{Robot}		 	&	Phase \#  & \textit{Trajectory} & \textit{Duration}\\
\hline 
}
$\rob{2}$	&	1					&	$\cycle(\alpha+t)$ 				& $\alpha + \beta + \delta_{(\alpha+\beta)/2}$ \\
				&	2					&	$
				\li{\cycle\left(\zeta_{\alpha,\beta} \right), 
					\cycle\left(\zeta_{\alpha,\beta}+\delta_{\beta/2} \right),
				  t}$ 			&  $\delta_{\beta/2}$ \\
				&	3					&	$\cycle\left(\zeta_{\alpha,\beta}+\beta+\delta_{\beta/2}+t\right)$ 									& $2\pi-2\alpha-2\beta - \delta_{(\alpha+\beta)/2}- \delta_{\beta/2}$ \\
\hline			
\end{tabular}
\end{center}
\end{definition}

Robots' trajectories $\alpha, \beta$ have the following meaning. As in the family of algorithms $\algo_{2}$, parameter $\alpha$ represents the arc distance the two robots have before the one preceding decides to jump ahead. In $\algo_{2}$ the two robots meet again once the jumper reaches the perimeter of the circle. In $\algo_{2}'$ the jumper deploys a little further away on the circle so that when the other robot reaches the deployment point of the jumper, the two robots are at arc distance $\beta$. As a result, the time it takes both robots to complete searching the entire circle is $2\pi-\alpha-\beta$,
as well as $\algo_2(\alpha,0)$ coincides with $\algo_2(\alpha)$. 
Finally, note that even though $\algo_{2}'$ will be eventually invoked for seemingly restricted values of $\beta$ ($\beta\leq \beta_0 \approx 0.04388$), the deviation in the performance will be significant enough (e.g. $\delta_{\beta_0/2}\approx 0.977997$) to account for its utilization in our upper bounds.

\section{Worst Case Performance Analysis}
\label{sec: new evac algos wrs}

In this section we perform worst case analysis for all algorithmic families $\algo_1, \algo_2, \algo_2'$ with respect to their parameters. 
Notably, results in this section are quantified formally and exactly by closed formulas. 
At a high level, each of $\algo_1, \algo_2, \algo_2'$ will be invoked to solve \evacw\ for different values of $w\in [\wrs{\ben{1}},\wrs{\ben{2}}]$, and each of them will have competitive average case performance for the corresponding worst case performance $w$. 
As an easy warm-up, we analyze $\algo_1$. 

\begin{lemma}[Worst Case Analysis for $\algo_1$]
\label{lem: algorithm new 1 wrst}
Let $\bar \alpha = 1+2\pi - w_{1}$, where $w_{1}= \wrs{\ben{1}}$. Then, for all $\alpha \in [0, \pi]$, we have that 
$$
\wrs{\algo_{1}(\alpha)} =
\left\{
\begin{array}{ll}
1+2\pi-\alpha&, \forall \alpha \in [0,\bar \alpha) \\
 \wrs{\ben{1}}&, \forall \alpha \in [\bar \alpha, \pi]
\end{array}
\right..
$$
\end{lemma}

%\ignore{
\begin{proof}[Lemma~\ref{lem: algorithm new 1 wrst}]
First it is easy to show that the worst case evacuation time is induced either when $\rob{1}$ finds the exit while moving from $\cycle(0)$ to $\cycle(\alpha)$, or while $\rob{1}, \rob{2}$ are exploring the circle together (after having met). 
By Lemma~\ref{lem: simple catching time}, the cost in the first case would be 
$$
\max_{0\leq x \leq \alpha}\{1+x+2\delta_x\}
=
\left\{
\begin{array}{ll}
1+\alpha+2\delta_\alpha&, \textrm{if}~\alpha\leq \alpha_0 \\
\wrs{\ben{1}}&, \textrm{otherwise}
\end{array}
\right.
$$ 
where the values of the piecewise function above follow from Lemma~\ref{lem: monotonicity ben1}. In the other case, the worst placement of exit is obtained using instances $\cycle(\alpha+\epsilon)$ for arbitrary small values of $\epsilon>0$ in which case the evacuation cost becomes $1+2\pi-\alpha$. 

Overall, is is easy to see that $1+\alpha_0 +2\delta_{\alpha_0} \leq 1+2\pi-\alpha_0$ showing that the dominant evacuation cost when $\alpha\leq \bar \alpha$ is $1+2\pi-\alpha$. For $\alpha>\bar \alpha$ the evacuation cost becomes equal to $w_{1}$. 
\qed
\end{proof}
%}

In a similar fashion, we can easily analyze $\algo_2$. 

\begin{lemma}[Worst Case Analysis for $\algo_2$]
\label{lem: algorithm new 2 wrst}
For all $\alpha \leq \pi-2$, we have 
$\wrs{\algo_{2}(\alpha)} =1+2\pi-\alpha$.
\end{lemma}

%\ignore{
\begin{proof}[Lemma~\ref{lem: algorithm new 2 wrst}]
%First we formally prove that $\wrs{\algo_{2}(\alpha)} =1+2\pi-\alpha$.
We distinguish three cases as to where the exit is. If $x\in [0,\alpha)$, then the worst instance $\cycle(x)$ is when $x=\alpha-\epsilon$ for arbitrarily small $\epsilon>0$, and the cost is $1+\alpha+2\delta_{\alpha/2}$. In the second case $x \in [\alpha, 2\alpha+\delta_{\alpha/2})$ and it is not difficult to see that the worst case induced cost in this case is not more than that of the first case. Finally, in the third case $x\in [ 2\alpha+\delta_{\alpha/2}, 2\pi)$, and the two robots move together, so the total cost, in the worst case, is $1+2\pi-\alpha$, 
when $x=2\pi-\epsilon$ for arbitrarily small $\epsilon>0$. It is not difficult to see that the dominant case is actually the third one, and in fact the two cases induce the same cost when 
%$1+2\pi-\alpha = \alpha+2\delta_{\alpha/2} $, 
$\pi=\alpha+\delta_{\alpha/2}$. By the definition of $\delta_{\alpha/2}$ we know that $\delta_{\alpha/2}=2\sinn{\frac{\alpha+\delta_{\alpha/2}}2}=2\sinn{\pi/2}=2$. Hence the costs become equal when $\alpha=\pi-2$. 
\qed
\end{proof}
%}

Next, we analyze $\algo_2'(\alpha,\beta)$, which requires more technical arguments. For this we will invoke $\algo_2'$ only for special parameters, whose choice is motivated by the following observation pertaining to the performance of $\algo_2$ (whose generalization is $\algo_2'$). 
From the proof of Lemma~\ref{lem: algorithm new 2 wrst}, it follows that among all algorithms $\algo_{2}(\alpha)$, where $\alpha\leq \gamma_0$ (see discussion before Definition~\ref{def: algo 2 new}), the one with minimum worst case evacuation cost is $\algo_{2}(\pi-2)$, and the cost becomes $3+\pi$. In fact, for all $w\in [3+\pi, 1+2\pi]$ there are two different values of $\alpha$ for which $\wrs{\algo_{2}(\alpha)}=w$, and we restrict $\alpha \in [0,\pi-2]$ so that we obtain evacuation algorithms with minimum average case cost. Moreover, $\alpha=\pi-2$ is the only parameter for which $\wrs{\algo_{2}(\alpha)}=3+\pi$ and as a byproduct, it is the algorithm in the family $\algo_{2}$ that minimizes the worst case. 

By Lemma~\ref{lem: algorithm new 2 wrst} we know that as $\beta\rightarrow 0$, the value of $\alpha$ that minimizes $\wrs{\algo_{2}'(\alpha,\beta)}$ approaches $\pi-2$. That value of $\alpha$ is what made the evacuation cost of $\algo_{2}(\alpha)$ attain the same value in two different (worst case) exit placements. Motivated by this, and for values of $\beta>0$ not too big, we still find the optimal choices of $\alpha$ that minimize the worst case performance. 

\begin{lemma}[Worst Case Analysis for $\algo_2'$]
\label{lem: wrs of algo 2}
Let $\beta_0=0.0438855$, and set $\alpha_\beta := \pi - \beta/2 - 2\coss{\beta/4}$. Then for all $\beta \in [0,\beta_0]$ we have 
$
\wrs{\algo_{2}'(\alpha_\beta, \beta)} = 1+\pi - \beta/2+2\coss{\beta/4}.
$
\end{lemma}

%\ignore{
\begin{proof}[Lemma~\ref{lem: wrs of algo 2}]
Let $w(\beta) = 1+\pi - \beta/2+2\coss{\beta/4}$. First we show that $w(\beta)$ is the worst case performance of $\algo_{2}'(\alpha_\beta, \beta)$ for two specific placements of the exit. 

We proceed by describing evacuation cost $\cost(x)$ assuming two arbitrary $\alpha,\beta$ for two different instances $\cycle(x)$. Using Lemma~\ref{lem: simple catching time}, we see that 
\begin{equation}
\label{equa: case 1 two jump wrst}
\lim_{\epsilon \rightarrow 0^+} \cost(\alpha-\epsilon) = 1+
\lim_{\epsilon \rightarrow 0^+} \se(\alpha-\epsilon)
+2\lim_{\epsilon \rightarrow 0^+} \ev(\alpha-\epsilon)
=
1+\alpha+2\delta_{\alpha/2}.
\end{equation}
Since the total search time is $2\pi-\alpha-\beta$, we also see that 
\begin{equation}
\label{equa: case 2 two jump wrst}
\lim_{\epsilon \rightarrow 0^+} \cost(2\pi-\epsilon)=1+2\pi-\alpha-\beta. 
\end{equation}
Now we claim that~\eqref{equa: case 1 two jump wrst},~\eqref{equa: case 2 two jump wrst} are equal when $\alpha=\alpha_\beta$. Indeed, equating~\eqref{equa: case 1 two jump wrst},~\eqref{equa: case 2 two jump wrst} gives 
\begin{equation}
\label{equa: a value 1}
a+\delta_{\alpha/2}=\pi-\beta/2. 
\end{equation}
But then, using~\eqref{equa: def of d_a}, we see that 
\begin{equation}
\label{equa: delta a/2 value 1}
\delta_{\alpha/2}=2\sinn{\frac{\alpha+\delta_{\alpha/2}}2}=2\sinn{\frac{\pi-\beta/2}2}=2\coss{\beta/4}.
\end{equation}
Substituting~\eqref{equa: delta a/2 value 1} into~\eqref{equa: a value 1}, we see that 
the value of $\alpha$ for which ~\eqref{equa: case 1 two jump wrst},~\eqref{equa: case 2 two jump wrst} are equal satisfies $\alpha= \pi - \beta/2 - 2\coss{\beta/4}$, as promised. Substituting this special value of $\alpha=\alpha_\beta$ either in~\eqref{equa: case 1 two jump wrst} or in~\eqref{equa: case 2 two jump wrst} induces evacuation cost $w(\beta)=1+\pi - \beta/2+2\coss{\beta/4}$.

Next we show that as long as $\beta$ is not too big, $w(\beta)$ is indeed the worst case evacuation cost. We consider the following cases $x\in I_i$, $i=1,\ldots,4$ for possible instances $\cycle(x)$; 
$I_1:=[0,\alpha ),
I_2:=[\alpha, 2\alpha+\beta+\delta_{(\alpha+\beta)/2} ),
I_3:=[2\alpha+\beta+\delta_{(\alpha+\beta)/2}, 2\alpha+2\beta+\delta_{(\alpha+\beta)/2}+\delta_{\beta/2} ),
I_4:=[2\alpha+2\beta+\delta_{(\alpha+\beta)/2}+\delta_{\beta/2}, 2\pi ). 
$
Clearly, \eqref{equa: case 1 two jump wrst}, \eqref{equa: case 2 two jump wrst} demonstrate the worst case evacuation costs for instances in $I_1, I_4$, respectively, and the cost in both cases, for $\alpha=\alpha_\beta$ is equal to $w(\beta)$. 

If $x\in I_2$ then $\cost(x)=1+\se(x)+2\ev(x)$. It is easy to see that both $\se(x), \ev(x)$ are monotone in $I_2$, so the worst case evacuation in this case is
\begin{equation}
\label{equa: 2jump case 2 cost}
\lim_{\epsilon \rightarrow 0^+} \cost(2\alpha_\beta+\beta+\delta_{(\alpha_\beta+\beta)/2}-\epsilon)=
1+\alpha_\beta+\beta+\delta_{(\alpha_\beta+\beta)/2}+2\delta_{\beta/2}. 
\end{equation}
Denote $\delta_{\beta/2}$ satisfying~\eqref{equa: def of d_a} by $\delta_\beta'$. 
Using~\eqref{equa: def of d_a} and the definition of $\alpha_\beta$, we see that 
$$
\delta_{(\alpha_\beta+\beta)/2}
=2\sinn{\frac{\alpha_\beta+\beta+\delta_{(\alpha_\beta+\beta)/2}}{2}}
= 2\coss{\coss{\beta/4}- \beta/4- \delta_{(\alpha_\beta+\beta)/2}}
$$
For simplicity, we denote $\delta_{(\alpha_\beta+\beta)/2}$ that satisfies the equation above by $\delta_\beta''$. Then, continuing from~\eqref{equa: 2jump case 2 cost}, the worst case evacuation cost when $x\in I_2$ becomes $1+\pi+\beta/2-2\coss{\beta/4}+\delta_\beta''+2\delta_\beta'$, an expression that depends exclusivey on $\beta$. 
The latter cost is no more than $w(\beta)$ if and only if $4\coss{\beta/4}-\beta-\delta_\beta''-2\delta_\beta' \geq 0$, and numerically we verify that this is satisfied as long as $\beta\leq \beta_0$ (see also Figure~\ref{fig: BoundOnBeta}). 
\begin{figure}[h!]
  \centering
  \includegraphics[width=0.5\linewidth]{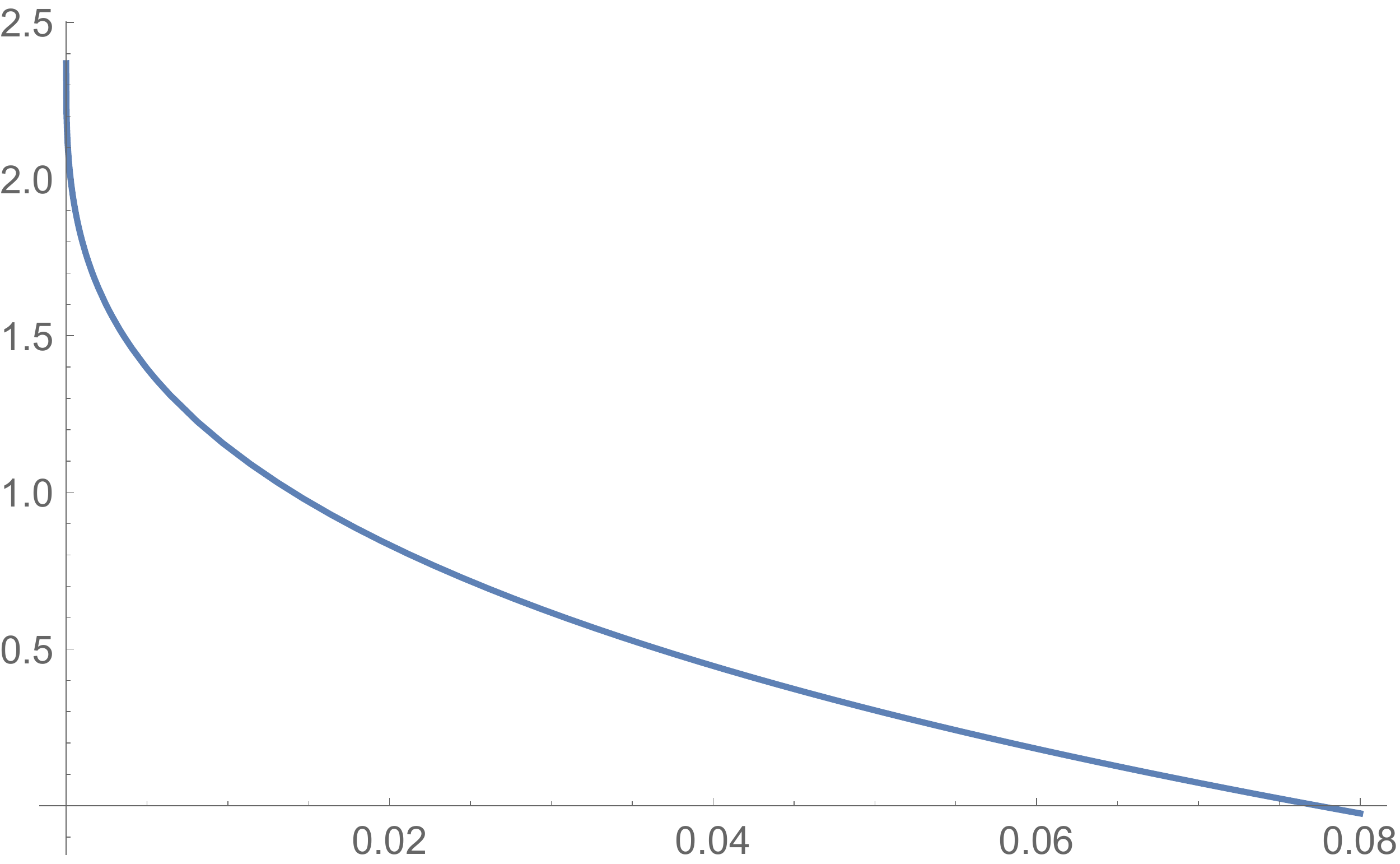}
\caption{The behavior of expression $4\coss{\beta/4}-\beta-\delta_\beta''-2\delta_\beta'$, for $\beta=0,\ldots, 0.8$. }
\label{fig: BoundOnBeta}
\end{figure}

Finally, it is easy to verify that $\delta_{\beta/2}$ and $|I_4|$ are increasing and decreasing respectively for $\beta\leq \beta_0$ and that $\delta_{\beta_0/2}=0.977997\leq 1.01099 = |I_4|$ (for $\beta=\beta_0$). As a result, the worst case evacuation cost of case $x\in I_3$ cannot exceed that of case $x\in I_4$, and hence the lemma follows. 
\qed
\end{proof}
%}

\section{Average Case Performance Analysis \& the Efficient Frontier}
\label{sec: new evac algos avg}

In this section we perform average case analysis for all algorithmic families $\algo_1, \algo_2, \algo_2'$, with respect to their parameters. For the sake of exposition of our results, we set $w_1=\wrs{\ben{1}}\approx 5.73906, w_2=\wrs{\ben{2}}=1+2\pi\approx 7.28319$ 
and for $\beta_0 \approx 0.04388$, as in Lemma~\ref{lem: wrs of algo 2}, we set $w_0:=\wrs{\algo_{2}'(\alpha_{\beta_0}, \beta_0)} \approx 6.11953$. 
%We also define $\alpha' :=\pi+\beta_0/2-\cos{\beta_0/4}\approx 1.16365$, and
We also recall $\bar \alpha\approx 1.54419$ of Lemma~\ref{lem: algorithm new 1 wrst}. Finally, we set
\begin{align*}
v(\alpha) 			&:= 1+2\pi-\alpha \\
v_2(\beta)			&:= 1+\pi - \beta/2+2\coss{\beta/4} \\
u_1(\alpha) 		&:= 
%0.0183446 \alpha^3 - 0.212852 \alpha^2 + 0.780803 \alpha + 4.19892    \\
0.00889 \alpha ^3-0.16944 \alpha ^2+0.71518 \alpha +4.23089\\
u_2'(\beta)		&:= 530.673 \beta^3 - 78.5498 \beta^2 + 7.36219 \beta + 4.70493  \\
u_2(\alpha)		&:= 0.093056 \alpha ^2+0.346659 \alpha +4.1719
%0.093056 \alpha^2+ 0.346659 \alpha + 4.1779
\end{align*}
Combined with our findings of Section~\ref{sec: new evac algos wrs}, the main result of the current section is the following. 

\begin{theorem}
\label{thm: efficient frontier}
For every $w\in [w_1,w_2]$ there is algorithm $\algo \in \{\algo_1, \algo_2',\algo_2\}$ and unique parameter(s) $p$ such that $\wrs{\algo(p)}=w$. In particular,  \\
- for all $\alpha\in [1, \bar \alpha]$, $\algo_1(\alpha)$ is $(u_1(\alpha), v(\alpha))$-efficient, 
%and $\textsc{image}_{\alpha\in [\alpha', \bar \alpha]} \{v(\alpha)\} =[w_1,2\pi]$, \\
and $v( [1, \bar \alpha] ) =[w_1,2\pi]$, \\
- for all $\beta\in [0,\beta_0] $, $\algo_2'(\alpha_\beta, \beta)$ is $(u_2'(\beta), v_2(\beta))$-efficient, 
%and $\textsc{image}_{\beta\in [0,\beta_0]} \{v_2(\beta)\} =[w_0,3+\pi]$, \\
and $v_2([0,\beta_0])=[w_0,3+\pi]$, \\
- for all $\alpha\in [0,\pi-2]$, $\algo_2(\alpha)$ is $(u_2(\alpha), v(\alpha))$-efficient, 
%and $\textsc{image}_{\alpha\in [0,\pi-2]} \{v(\alpha)\} =[3+\pi,w_2]$.
and $v( [0,\pi-2])=[3+\pi,w_2]$.
\end{theorem}

%\ignore{
\begin{proof}[Theorem~\ref{thm: efficient frontier}]
The claims for the worst case performances of $\algo_1, \algo_2', \algo_2$ follow directly from Lemmata~\ref{lem: algorithm new 1 wrst},~\ref{lem: wrs of algo 2} and~\ref{lem: algorithm new 2 wrst}, respectively. Next we argue that as the parameters vary in their specified range, we obtain the entire spectrum of $w\in [w_1,w_2]$, and this for unique values of the parameters. For this, we will rely on that for all evacuation algorithm families, the worst case cost is monotone in the parameters.%, and in particular both $v(\alpha), v_2(\beta)$ are strictly decreasing in $\alpha,\beta$, respectively. 

First, we argue about $\algo_1$. We observe that by the definition of $\bar \alpha$, $\wrs{\algo_1(\bar \alpha)}=w_1$, 
and  $\wrs{\algo_1(1)}=1+2\pi-1=2\pi$. Together with the fact that $v(\alpha)$ is strictly decreasing, we see that $\wrs{\algo_1(\alpha)}$ is 1-1 and onto to $[w_1,2\pi]$ as $\alpha$ ranges in $[1, \bar \alpha]$. 

Second, we study $\algo_2'$ whose worst case cost $v_2(\beta)$ is strictly decreasing in $\beta$. Moreover, by definition of $\beta_0$, we have $\wrs{\algo_2(\alpha_{\beta_0},\beta_0)}=w_0$. Then we note that for $\beta = 0$, $\algo_2(\alpha_\beta,\beta)$ coincides with $\algo_2(\pi-2)$, and in particular the induced worst case cost becomes $3+\pi$. Therefore 
$\wrs{\algo_2'(\alpha_\beta,\beta)}$ is 1-1 and onto to $[w_0,3+\pi]$ as $\beta$ ranges in $[0, \beta_0]$.
%$\textsc{image}_{\beta\in [0,\beta_0]} \{v_2(\beta)\} =[w_0,3+\pi]$, and again, for each $w$ in that range there exists unique $\beta \in [0,\beta_0$ inducing that worst case cost. 

Third, we study $\algo_2$, for which we know that $\wrs{\algo_2(\pi-2)}=3+\pi$. Again, the worst case cost is monotone in $\alpha$ and $\algo_2(0)$ coincides with benchmark algorithm $\ben{2}$, that is $\wrs{\algo_2(0)}=w_2$. Hence, 
$\wrs{\algo_2(\alpha)}$ is 1-1 and onto to $[3+\pi,w_2]$ as $\alpha$ ranges in $[0,\pi-2]$. 
%for each $w\in [3+\pi, w_2]$ there is unique $\alpha \in [0,\pi-2]$ giving worst case performance $w$. 

Finally, we argue that 
\begin{align*}
\avg{\algo_{1}(\alpha)} \leq u_1(\alpha) 					&, \forall \alpha \in  [1, \bar \alpha]\\
\avg{\algo_{2}'(\alpha_\beta,\beta)} \leq u_2'(\beta)	&, \forall \beta \in  [0, \beta_0]\\
\avg{\algo_{2}(\alpha)} \leq u_2(\alpha)					&, \forall \alpha \in  [0, \pi-2]
\end{align*}
For this, we numerically compute $\avg{\algo_{1}(\alpha)}, \avg{\algo_{2}'(\alpha_\beta,\beta)}, \avg{\algo_{2}(\alpha)}$
for various values of parameters $\alpha,\beta$, and we heuristically choose $u_1, u_2',u_2$ so as to upper bound the average case performance of $\algo_1, \algo_2', \algo_2$, effectively verifying our claim numerically. 
For each evacuation algorithm, we utilize Corollary~\ref{cor: cost simplified}, which together with the analytic description of our evacuation algorithms (see Definitions~\ref{def: algo 1 new},~\ref{def: algo 2 prime new}, and~\ref{def: algo 2 new})
allow us to compute their average case performance using computer-assisted calculations. 
Our numerical calculations are depicted in 
Figure~\ref{fig: avg upper bounds}.

\begin{figure}[h!]
\centering
  \includegraphics[width=.31\linewidth]{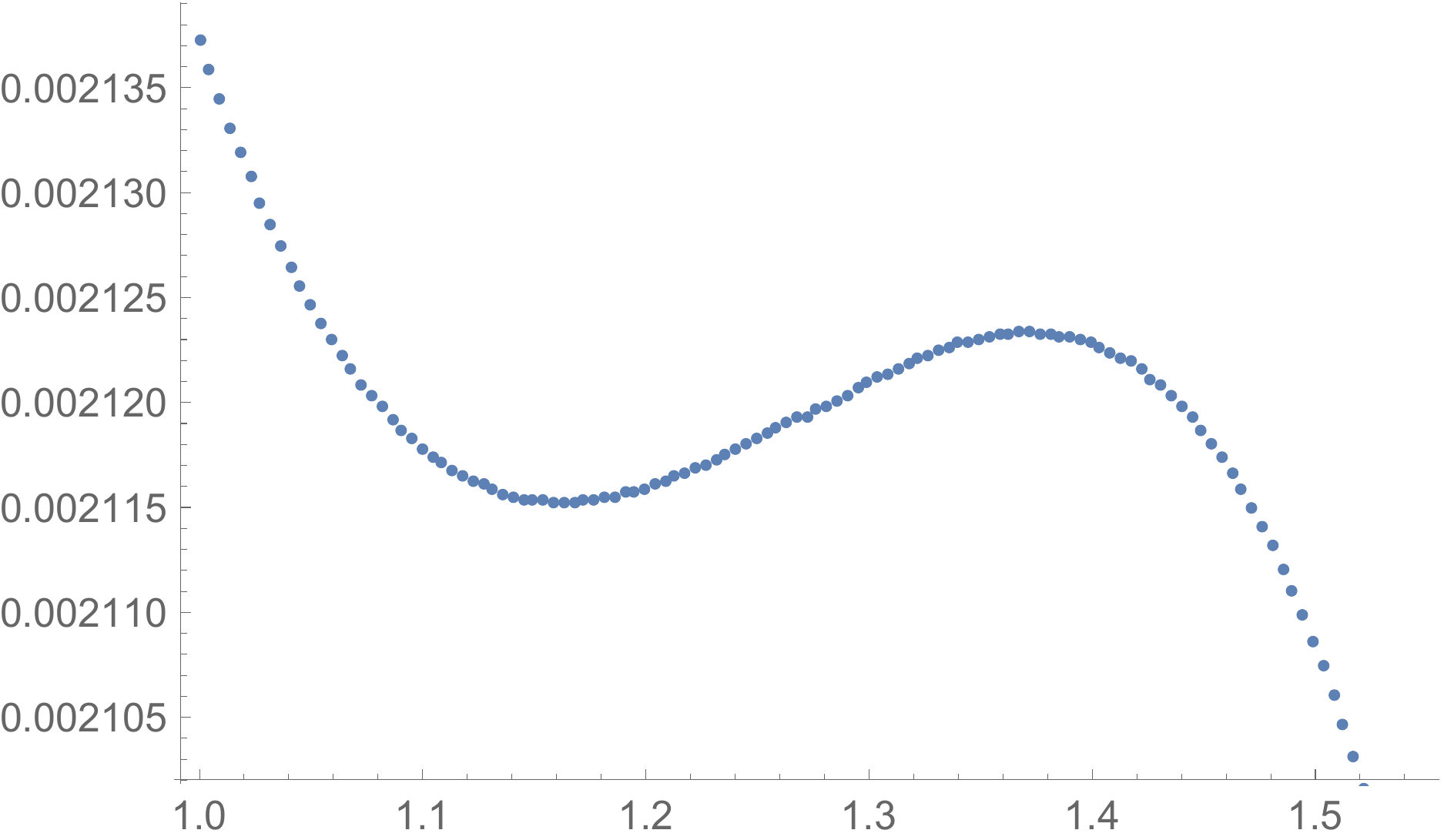}	
  %see % for algo 1 see Algo1-AvgUpperBound- Unified Approach -Algo7b-March27-b-zero.nb
  %see % for algo 1 new see Algo1-AvgUpperBound- Unified Approach -Algo7b-March27-b-zeroNEWRangeOfAlphas.nb
  ~~~ \includegraphics[width=.31\linewidth]{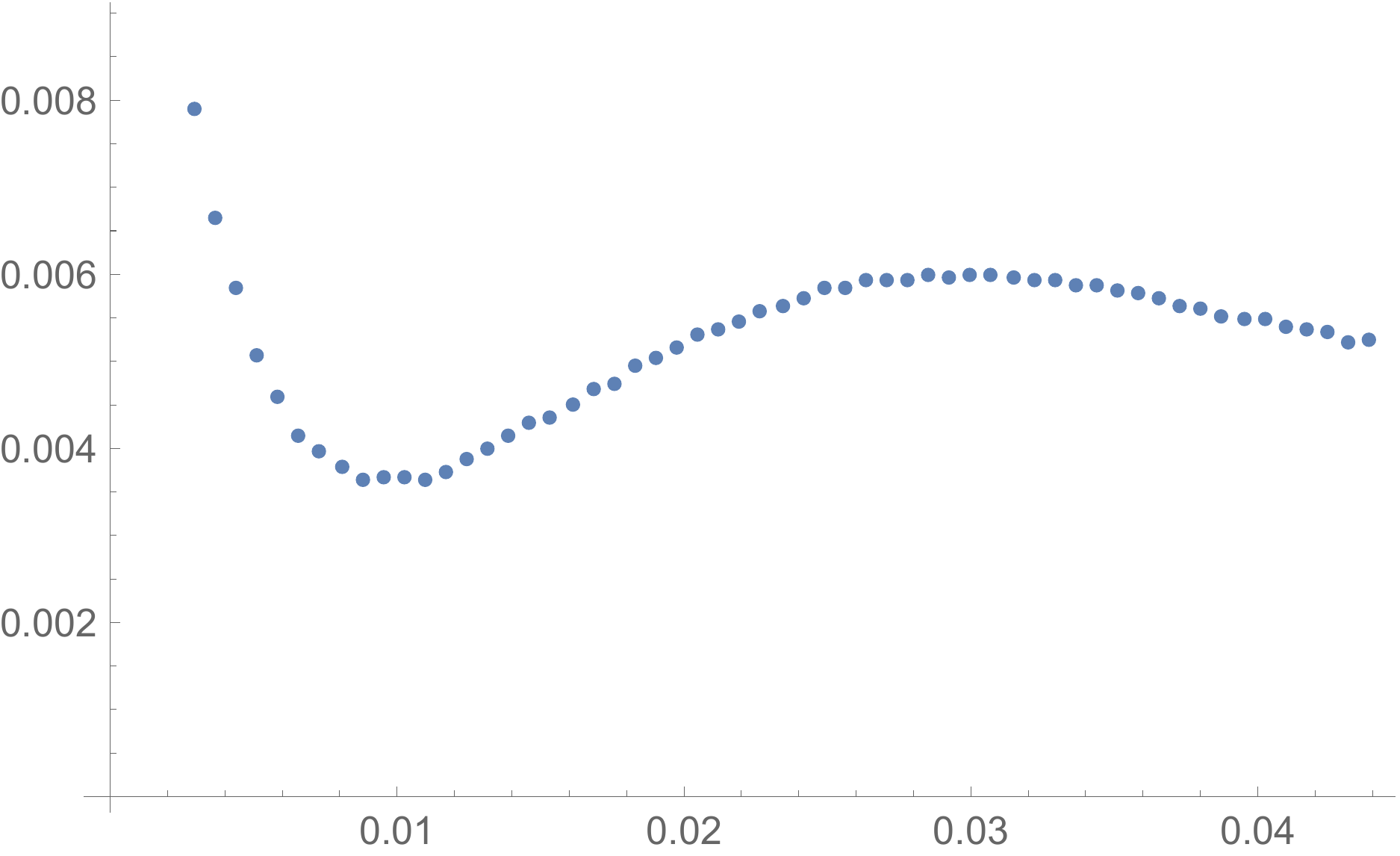}
~~~~  \includegraphics[width=.31\linewidth]{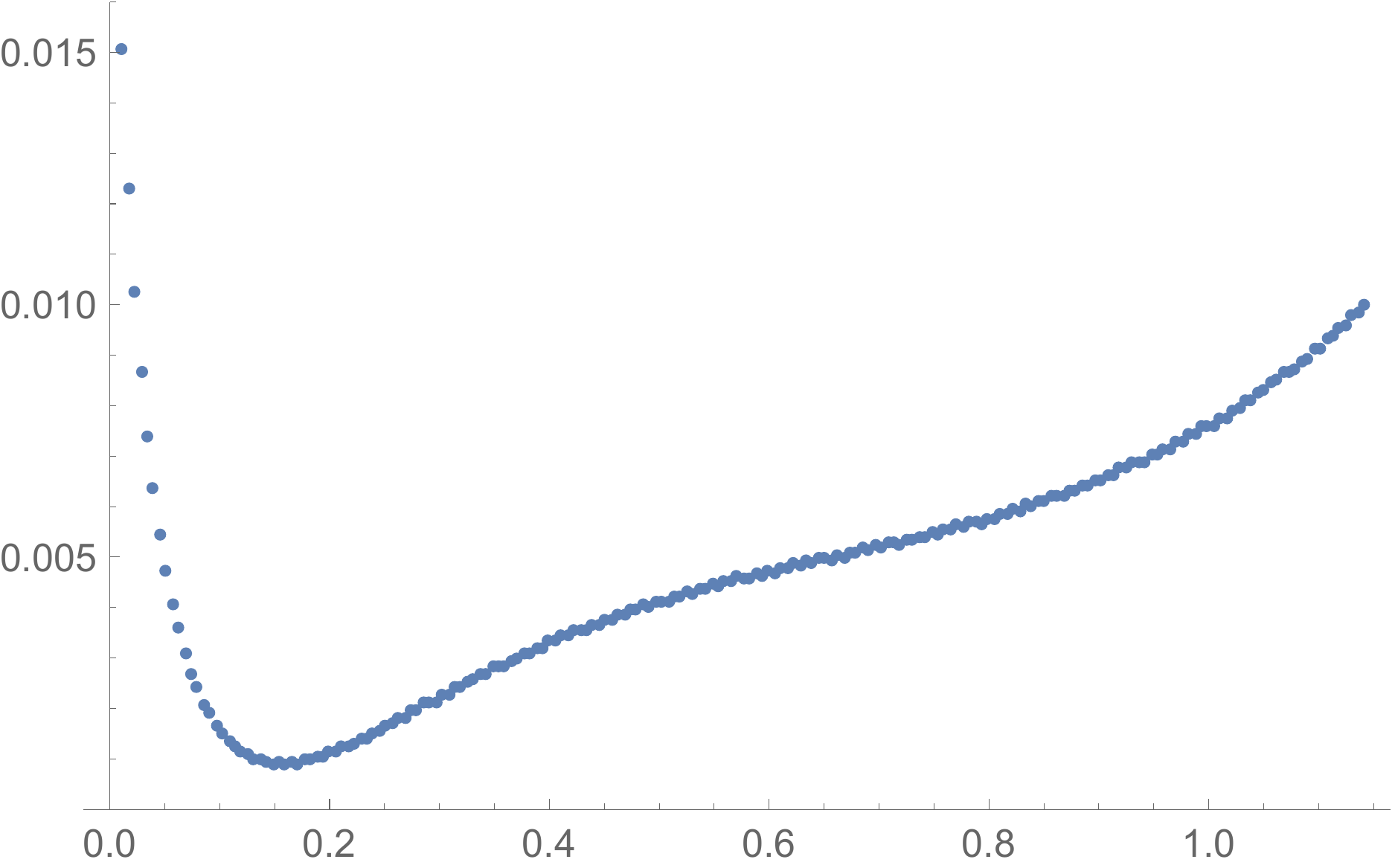}	%see Algo2-AvgUpperBound.nb
\caption{
On the right $u_1(\alpha)-\avg{\algo_{1}(\alpha)}$, for $\alpha'\leq \alpha \leq \bar \alpha$. 
In the middle, $u_2'(\beta)-\avg{\algo_{2}'(\alpha_\beta, \beta)}$, for $0\leq \beta \leq \beta_0$.
On the right $u_2(\alpha)-\avg{\algo_{2}(\alpha)}$, for $0\leq \alpha \leq \pi-2$.
 }
\label{fig: avg upper bounds}
\end{figure}
\qed
\end{proof}
%}

Finally, we aim to formally quantify the efficient frontier of our algorithms as depicted in Figure~\ref{fig: entire efficient frontier} (see Section~\ref{sec: definition}). 
The parametric curves described in Theorem~\ref{thm: efficient frontier} provide, strictly speaking, an upper bound for the parametric curve of Figure~\ref{fig: entire efficient frontier}. Next, we compute $g:\reals\mapsto\reals$, so that the parametric curves of Theorem~\ref{thm: efficient frontier} are written in the form $\{(g(w), w)\}_{w\in [w_1,w_2]}$. That would also imply that there is a solution to \evacw\ of cost at most $g(w)$. 

In that direction, we study each evacuation algorithm family $\algo(p)$ with worst case performance, say, $v(p)$, and average case upper bound, say, $u(p)$. For each $w \in [w_1,w_2]$ in the range of $\algo(p)$, we set $p=v^{-1}(w)$ so that the average case performance achieved becomes $u(v^{-1}(w))$.

Recall that $\wrs{\algo_i(\alpha)}=v(\alpha)$, so that $v^{-1}(w)=1+2\pi-w$, and hence for algorithms $\algo_i$ we can easily compute $u_i(v^{-1}(w))$, $i=1,2$. 
For $\algo_2'$ we recall that $\avg{\algo_2'(\alpha_\beta,\beta)}$ is decreasing in $\beta$. Since $v_2^{-1}$ does not admit a closed form, we need to observe that $2.999+\pi-\beta/2 \leq v_2(\beta)\leq 3+\pi-\beta/2$ for all $\beta \in [0,\beta_0]$ so that an upper bound for $\avg{\algo_2'(\alpha_\beta,\beta)}$ admitting worst case performance $w$ can be computed by 
$u_2'(12.2812 -2w)$. 

Now for each $w \in [w_1,w_2]$ we need to specify which of the evacuation algorithms we will invoke. Note that in Theorem~\ref{thm: efficient frontier} we chose the range of $\alpha$ in $\algo_1$ to start from $1$ so that as to guarantee that  $\wrs{\algo_1(1)}\geq w_0$. We note that $u_2'(12.2812 -2w) = u_1(1+2\pi-w)$ for $w'\approx6.12851$, so algorithm $\algo_1$ should be invoked for $w\in [w_1,w']$ (and $w'$ is obtained for $\alpha':=1+2\pi-w'\approx 1.15468$), then $\algo_2'$ for $w\in [w',3+\pi]$ (and $w'$ is obtained for $\beta'$ so that $v_2(\beta') = w'$, where $\beta'\approx 0.0241653$), and $\algo_2$ for $w\in [3+\pi,w_2]$. 
We conclude with the next Theorem (for convenience, the values of all constants are summarized at the end of Section~\ref{sec: trajectories description}). 

\begin{theorem}
\label{thm: main thm upper bound}
%Let $w_{1}, w_{2}$ be the worst case performance of Benchmark Algorithms $\ben{1}$ and $\ben{2}$, respectively. 
For every $w \in [w_1,w_2]$, the optimal solution to \evacw\ is at most $g(w)$, where 
$$
g(w) = 
\left\{
\begin{array}{ll}
-0.00889 w^3+0.0248026 w^2+0.338241 w+3.88629
&, w\in [w_1,w'] ~~(\algo_{1}(\alpha), ~~\alpha \in [\alpha', \bar \alpha]) \\
-4245.38 w^3+77893.3 w^2-476397. w+971235
&, w\in [w',3+\pi] ~~(\algo_{2}(\alpha_\beta, \beta), ~~\beta \in [0, \beta']) \\
0.093056 w^2-1.70215 w+11.6328
&, w\in [3+\pi,w_2] ~~(\algo_{2}(\alpha), ~~\alpha \in [0, \pi-2]) 
\end{array}
\right.
$$ 
\end{theorem}

\ignore{
$
w_1 \approx 5.73906,
w_0 \approx 6.11953,
w' \approx 6.12851,
w_2 \approx 7.28319,
\alpha_0 \approx 0.96782,
\alpha' \approx 1.15468,
\bar \alpha \approx 1.54419,
\beta' \approx 0.0241653,
\beta_0 \approx 0.04388.
$
}

\section{Conclusion \& Open Problems}
\label{sec: conclusion}

Our work suggests a number of open problems directly aiming to understand \evacw\ better. Apart from generally improving our upper bounds, we find the following list of questions particularly interesting and challenging: 
\begin{enumerate}[(a)]
\item Note that when $w=\wrs{\ben{1}}$, we presented algorithm $\algo_1(\alpha)$ which, for certain value of $\alpha$, has worst case performance equal to $w$ and average case performance less that $\avg{\ben{1}}$. Is there an algorithm whose average case performance is no more than $\avg{\ben{1}}$, and worst case performance strictly less than $w$?
\item  Is it true that the best possible efficient frontier is given by a smooth transition between families of evacuation algorithms? Note that $\algo_2$ naturally extends $\ben{2}$, $\algo_2'$ naturally extends  $\algo_2$, and that $\algo_1$ naturally extends $\ben{1}$. However, $\algo_1$ and $\algo_2'$ behave differently, even though their efficient frontier agrees for certain values of the parameters. 
\item $\avg{\ben{2}}=1+\pi$, and none of our algorithms beat this performance. We conjecture that this is the best possible average evacuation time, even in the wireless model, and for any number of robots.
\end{enumerate}
Apart from the list above, we believe that the direction of studying randomized algorithms for evacuation-type problems, especially with respect to average case/worst case trade-offs is of special interest, and should be considered for existing as well as for new search problems in the area. 

\bibliography{refs}

\begin{thebibliography}{10}

\bibitem{ahlswede1987search}
R.~Ahlswede and I.~Wegener.
\newblock {\em Search problems}.
\newblock Wiley-Interscience, 1987.

\bibitem{AH00}
S.~Albers and M.~R. Henzinger.
\newblock Exploring unknown environments.
\newblock {\em SIAM Journal on Computing}, 29(4):1164--1188, 2000.

\bibitem{AKS02}
S.~Albers, K.~Kursawe, and S.~Schuierer.
\newblock Exploring unknown environments with obstacles.
\newblock {\em Algorithmica}, 32(1):123--143, 2002.

\bibitem{alpern2002theory}
S.~Alpern and S.~Gal.
\newblock {\em The theory of search games and rendezvous}, volume~55.
\newblock Kluwer Academic Publishers, 2002.

\bibitem{Alpern2013}
Steve Alpern, Robbert Fokkink, Leszek Gasieniec, Roy Lindelauf, and V.S.
  Subrahmanian, editors.
\newblock {\em Ten Open Problems in Rendezvous Search}, pages 223--230.
\newblock Springer NY, New York, NY, 2013.

\bibitem{baezayates1993searching}
R.~Baeza~Yates, J.~Culberson, and G.~Rawlins.
\newblock Searching in the plane.
\newblock {\em Information and Computation}, 106(2):234--252, 1993.

\bibitem{baumann2009earliest}
Nadine Baumann and Martin Skutella.
\newblock Earliest arrival flows with multiple sources.
\newblock {\em Mathematics of Operations Research}, 34(2):499--512, 2009.

\bibitem{beck1964linear}
A.~Beck.
\newblock On the linear search problem.
\newblock {\em Israel J. of Mathematics}, 2(4):221--228, 1964.

\bibitem{bellman1963optimal}
R.~Bellman.
\newblock An optimal search.
\newblock {\em SIAM Review}, 5(3):274--274, 1963.

\bibitem{benkoski1991survey}
S.~Benkoski, M.~Monticino, and J.~Weisinger.
\newblock A survey of the search theory literature.
\newblock {\em Naval Research Logistics (NRL)}, 38(4):469--494, 1991.

\bibitem{Watten2017}
S.~Brandt, F.~Laufenberg, Y.~Lv, D.~Stolz, and R.~Wattenhofer.
\newblock Collaboration without communication: Evacuating two robots from a
  disk.
\newblock In {\em Proceedings of Algorithms and Complexity - 10th International
  Conference, {CIAC} 2017, Athens, Greece, May 24-26, 2017}, pages 104--115,
  2017.

\bibitem{B05}
W.~Burgard, M.~Moors, C.~Stachniss, and F.~E. Schneider.
\newblock Coordinated multi-robot exploration.
\newblock {\em Robotics, IEEE Transactions on}, 21(3):376--386, 2005.

\bibitem{Groupsearch}
M.~Chrobak, L.~Gasieniec, Gorry T., and R.~Martin.
\newblock Group search on the line.
\newblock In {\em SOFSEM 2015}. Springer, 2015.

\bibitem{chung2011search}
T.~H. Chung, G.~A. Hollinger, and V.~Isler.
\newblock Search and pursuit-evasion in mobile robotics.
\newblock {\em Autonomous robots}, 31(4):299--316, 2011.

\bibitem{DBLP:conf/icdcn/CzyzowiczDGKM16}
J.~Czyzowicz, S.~Dobrev, K.~Georgiou, E.~Kranakis, and F.~MacQuarrie.
\newblock Evacuating two robots from multiple unknown exits in a circle.
\newblock {\em Theor. Comput. Sci.}, 709:20--30, 2018.

\bibitem{CGGKMP}
J.~Czyzowicz, L.~Gasieniec, T.~Gorry, E.~Kranakis, R.~Martin, and D.~Pajak.
\newblock Evacuating robots via unknown exit in a disk.
\newblock In {\em Proceedings DISC, Austin, Texas}, pages 122--136. Springer,
  2014.

\bibitem{georgioudiskfaulty2017}
J.~Czyzowicz, K.~Georgiou, M.~Godon, E.~Kranakis, D.~Krizanc, W.~Rytter, and
  M.~Wlodarczyk.
\newblock Evacuation from a disc in the presence of a faulty robot.
\newblock In {\em Proceedings SIROCCO 2017, 19-22 June 2017, Porquerolles,
  France}, pages 158--173, 2018.

\bibitem{CGKKKNOS18a}
J.~Czyzowicz, K.~Georgiou, R.~Killick, E.~Kranakis, D.~Krizanc, L.~Narayanan,
  J.~Opatrny, and S.~Shende.
\newblock God save the queen.
\newblock In {\em 9th International Conference on Fun with Algorithms
  ({FUN'18})}, 2018.

\bibitem{CGKKKNOS18b}
J.~Czyzowicz, K.~Georgiou, R.~Killick, E.~Kranakis, D.~Krizanc, L.~Narayanan,
  J.~Opatrny, and S.~Shende.
\newblock Priority evacuation from a disk using mobile robots.
\newblock In {\em 25th International Colloquium on Structural Information and
  Communication Complexity, June 18-21, 2018, Ma'ale HaHamisha, Israel},
  Manuscript, 2018.

\bibitem{isaacCzyzowiczGKKNOS16}
J.~Czyzowicz, K.~Georgiou, E.~Kranakis, D.~Krizanc, L.~Narayanan, J.~Opatrny,
  and S.~Shende.
\newblock Search on a line by byzantine robots.
\newblock In {\em 27th International Symposium on Algorithms and Computation,
  {ISAAC} 2016, December 12-14, 2016, Sydney, Australia}, pages 27:1--27:12,
  2016.

\bibitem{CGKNOV}
J.~Czyzowicz, K.~Georgiou, E.~Kranakis, L.~Narayanan, J.~Opatrny, and
  B.~Vogtenhuber.
\newblock Evacuating robots from a disc using face to face communication.
\newblock In {\em CIAC 2015}. Springer, 2015.

\bibitem{CzyzowiczGKNOV15}
J.~Czyzowicz, K.~Georgiou, E.~Kranakis, L.~Narayanan, J.~Opatrny, and
  B.~Vogtenhuber.
\newblock Evacuating robots from a disk using face-to-face communication
  (extended abstract).
\newblock In {\em Proceedings of Algorithms and Complexity, {CIAC} 2015, Paris,
  France, May 20-22, 2015}, pages 140--152, 2015.

\bibitem{DKP91}
X.~Deng, T.~Kameda, and C.~Papadimitriou.
\newblock How to learn an unknown environment.
\newblock In {\em FOCS}, pages 298--303. IEEE, 1991.

\bibitem{dobbie1968survey}
J.~Dobbie.
\newblock A survey of search theory.
\newblock {\em Operations Research}, 16(3):525--537, 1968.

\bibitem{FGK10}
S.~Fekete, C.~Gray, and A.~Kr{\"o}ller.
\newblock Evacuation of rectilinear polygons.
\newblock In {\em Combinatorial Optimization and Applications}, pages 21--30.
  Springer, 2010.

\bibitem{FT08}
F.~V. Fomin and D.~M. Thilikos.
\newblock An annotated bibliography on guaranteed graph searching.
\newblock {\em Theoretical Computer Science}, 399(3):236--245, 2008.

\bibitem{GeorgiouKK16}
Konstantinos Georgiou, George Karakostas, and Evangelos Kranakis.
\newblock Search-and-fetch with one robot on a disk - (track: Wireless and
  geometry).
\newblock In {\em Algorithms for Sensor Systems - 12th International Symposium
  on Algorithms and Experiments for Wireless Sensor Networks, {ALGOSENSORS}
  2016, Aarhus, Denmark, August 25-26, 2016, Revised Selected Papers}, pages
  80--94, 2016.

\bibitem{GeorgiouKK17}
Konstantinos Georgiou, George Karakostas, and Evangelos Kranakis.
\newblock Search-and-fetch with 2 robots on a disk - wireless and face-to-face
  communication models.
\newblock In Federico Liberatore, Greg~H. Parlier, and Marc Demange, editors,
  {\em Proceedings of the 6th International Conference on Operations Research
  and Enterprise Systems, ICORES 2017, Porto, Portugal, February 23-25, 2017},
  pages 15--26. SciTePress, 2017.

\bibitem{georgiou2017searching}
Kostantinos Georgiou, Evangelos Kranakis, and Alexandra Steau.
\newblock Searching with advice: Robot fence-jumping.
\newblock {\em Journal of Information Processing}, 25:559--571, 2017.

\bibitem{HIKK01}
F.~Hoffmann, C.~Icking, R.~Klein, and K.~Kriegel.
\newblock The polygon exploration problem.
\newblock {\em SIAM Journal on Computing}, 31(2):577--600, 2001.

\bibitem{K94}
J.~Kleinberg.
\newblock On-line search in a simple polygon.
\newblock In {\em SODA}, page~8. SIAM, 1994.

\bibitem{lidbetter2013hide}
Thomas Lidbetter.
\newblock {\em Hide-and-seek and other search games}.
\newblock PhD thesis, The London School of Ecoomics and Political Science
  (LSE), 2013.

\bibitem{mitchell2000geometric}
Joseph~SB Mitchell.
\newblock Geometric shortest paths and network optimization.
\newblock {\em Handbook of computational geometry}, 334:633--702, 2000.

\bibitem{nahin2012chases}
P.~Nahin.
\newblock {\em Chases and Escapes: The Mathematics of Pursuit and Evasion}.
\newblock Princeton University Press, 2012.

\bibitem{PY}
C.~H Papadimitriou and M.~Yannakakis.
\newblock Shortest paths without a map.
\newblock In {\em ICALP}, pages 610--620. Springer, 1989.

\bibitem{stone1975theory}
L.~Stone.
\newblock {\em Theory of optimal search}.
\newblock Academic Press New York, 1975.

\bibitem{T01}
S.~Thrun.
\newblock A probabilistic on-line mapping algorithm for teams of mobile robots.
\newblock {\em The International Journal of Robotics Research}, 20(5):335--363,
  2001.

\bibitem{Y98}
B.~Yamauchi.
\newblock Frontier-based exploration using multiple robots.
\newblock In {\em Proceedings of the second international conference on
  Autonomous agents}, pages 47--53. ACM, 1998.

\end{thebibliography}
\bibliographystyle{plain}

\end{document}